\documentclass[12pt]{article}

\usepackage[margin=1in]{geometry}
\usepackage{caption}

\usepackage{hyperref}
\usepackage{amsmath}
\usepackage{graphicx, psfrag, epsf}
\usepackage{enumerate}
\usepackage{natbib}
\usepackage{url} 

\usepackage{xr}

\usepackage{amssymb, amsfonts, amsthm, bm}
\usepackage{mathtools, mathrsfs}
\usepackage{relsize} 
\usepackage{hyperref}
\usepackage{algorithm}
\usepackage{multirow}
\usepackage{comment}
\usepackage{caption}

\newtheorem{theorem}{Theorem}
\newtheorem{lemma}{Lemma}

\newtheorem{corollary}{Corollary}

\newtheorem{example}{Example}
\newtheorem{assumption}{Assumption}

\def\beginmat{ \left( \begin{array} }
\def\endmat{ \end{array} \right) }

\def\log{{\rm log}}
\def\tr{{\rm tr}}

\DeclareMathOperator{\E}{\mathbb{E}}

\DeclareMathOperator{\Cov}{Cov}
\DeclareMathOperator{\Length}{length}

\DeclareMathOperator{\MN}{\mathrm{MN}}

\usepackage{color,soul}
\soulregister\cite7
\soulregister\citep7
\soulregister\eqref7
\soulregister\pageref7

\newcommand{\etr}{\mathrm{etr}}

\newcommand{\beginsupplement}{
\setcounter{table}{0}
\renewcommand{\thetable}{S\arabic{table}}
\setcounter{figure}{0}
\renewcommand{\thefigure}{S\arabic{figure}}
\setcounter{equation}{0}
\renewcommand{\theequation}{S\arabic{equation}}
\setcounter{lemma}{0}
\renewcommand{\thelemma}{S\arabic{lemma}}
\setcounter{theorem}{0}
\renewcommand{\thetheorem}{S\arabic{theorem}}
\setcounter{section}{0}
\renewcommand{\thesection}{S\arabic{section}}
}

 \newcommand{\noop}[1]{}

\title{Gaussian orthogonal latent factor processes for large incomplete matrices of correlated data}
\providecommand{\keywords}[1]
{
  \textbf{\textit{Keywords---}} #1
}

\date{}

\author{Mengyang Gu\thanks{Department of Statistics and Applied Probability, University of California, Santa Barbara, CA, Email: \href{mailto:mengyang@pstat.ucsb.edu}{mengyang@pstat.ucsb.edu}} \, and Hanmo Li\footnote{
Department of Statistics and Applied Probability, University of California, Santa Barbara, CA 93106, USA, Email: \href{mailto:hanmo@pstat.ucsb.edu}{hanmo@pstat.ucsb.edu} }} 

\begin{document}
\maketitle

	\begin{abstract}
	We introduce  Gaussian orthogonal latent factor processes for modeling and predicting large correlated data.  To handle the computational challenge, we first decompose the likelihood function of the Gaussian random field with 
	{{a}} multi-dimensional input domain into a product of densities at the orthogonal components with {{lower-dimensional}} inputs. The continuous-time Kalman filter is implemented to compute the likelihood function {{efficiently}} without making {{approximations}}.  We also show that the posterior distribution of the factor processes {{is}} independent, as a consequence of prior independence of factor processes and orthogonal factor loading matrix. For studies with large sample {{sizes}}, we propose a flexible way to model the mean, and we derive the marginal posterior distribution to solve identifiability issues in sampling these parameters. Both simulated and real data applications confirm the outstanding performance of this method.  
	\end{abstract}

~~


	\section{Introduction}
	Large spatial, spatio-temporal, and functional data are commonly used in various studies, including geological hazard quantification,  engineering,  and medical imaging, to facilitate scientific discoveries. Many data sets are observed on  incomplete matrices  with missing values due to the limitation of {{the}} technique or computational cost. 
	
	Gaussian processes (GPs)  are widely used for modeling correlated data  \citep{banerjee2014hierarchical,cressie1993statistics}. 	Computing the likelihood function from a GP model, however, generally takes $O(N^3_o)$ operations in finding the inverse and determinant of the covariance matrix, where $N_o$ is the number of observations. The computational bottleneck prevents modeling a large correlated data set by GPs directly. Tremendous efforts have been made to approximate a GP model in recent studies,  including, for example, stochastic partial differential equation approach \citep{lindgren2011explicit,rue2009approximate}, hierarchical nearest neighbor methods \citep{datta2016hierarchical}, multi-resolution process \citep{katzfuss2017multi}, local Gaussian process approach \citep{gramacy2015local},  periodic embedding \citep{guinness2017circulant,stroud2017bayesian} and covariance tapering  \citep{kaufman2008covariance}, which have obtained wide attention in recent years. 

	Compared to a large number of studies on approximating GPs, less progress have been made on efficiently computing the likelihood function without approximation. 
	In this work, we propose a flexible and computationally feasible approach to model large incomplete matrix observations of correlated data, called  Gaussian orthogonal latent factor (GOLF) processes. 
    Bayesian inference was derived to assess the uncertainty in parameter estimation and predictions. 
	 GPs with product covariance functions on lattice observations or semiparametric latent factor models  \citep{sacks1989design,kennedy2001bayesian,seeger2005semiparametric} can be represented as full-rank GOLF processes, which permit much smaller computational costs than directly computing the likelihood function and making  predictions. Further reducing the computational cost can be achieved by low-rank  GOLF processes, where the computational cost {{is}} similar to the order of principal component analysis. 




	
	

	We highlight a few contributions of this work.
	We first show that for GPs with product  covariance functions or semiparametric latent factor models, if the latent factor loading matrix is orthogonal,  prior independence of latent factor processes implies posterior independence of factor processes. The new finding allows one to decompose the likelihood function of lattice data into a product of  densities of projected output, which greatly reduces the computational complexity. Separate continuous-time Kalman filters can be applied to compute the posterior distributions of  factor processes at lower dimensional inputs in parallel, which has linear computational operations with respect to the number of observations.
	Second, as {{a}} large number of observations {{provide}} rich information,  we introduce a flexible  way to model the mean function and derive the marginal posterior distribution of the linear coefficients,  
	to solve identifiability issues in posterior sampling. Furthermore, compared with the maximum marginal likelihood estimation of factor loadings derived in  \cite{gu2018generalized},	our approach is applicable to model    observations on incomplete lattice. Finally, we developed Bayesian inference for  uncertainty assessment, which {{is}} critically important for  inverse problems in applications  \citep{kennedy2001bayesian,bayarri2007framework}.

 The purpose of this work are twofold. First, we aim to develop a pipeline of computationally efficient  methods of modeling correlated data with multi-dimensional input without approximating the likelihood function.  
  Properties of GOLF processes derived in this work are useful for developing an efficient approximation algorithm for scenarios with  multi-dimensional input variables. 
 Besides, the nonseparable covariance and coordinate-specific mean coefficients proposed in this work provide flexible choices for models of local information. 
 Second, we primarily focus on applications based on images,  which include  inverse problems by satellite radar interferograms \citep{anderson2019magma}, and  estimating dynamic information from microscopic videos  \citep{cerbino2008differential}. Our approach allows for efficient Bayesian inference in a large sample scenario.


   The rest of the article is organized as follows. In Section \ref{subsec:orthogonal_independence}, we introduce the GOLF model  with an emphasis on {{the}} orthogonal decomposition of the likelihood function and posterior independence of latent factor processes. The flexible mean function, spatial latent factor loading matrix and kernel functions are discussed in Section \ref{subsec:mean_function}-\ref{subsec:kernel}, respectively. We introduce the Markov Chain Monte Carlo (MCMC) algorithm and discuss the computational complexity in Section \ref{subsec:MCMC}. In Section \ref{subsec:KF}, we introduce the continuous-time Kalman filter in computing the likelihood function with linear computational complexity. Section \ref{sec:comparison} compares our approach with other alternatives, and numerical results for comparing these approaches are 
   presented in Section \ref{sec:simulation}-\ref{sec:real_data}. We conclude this work and discuss  several potential extensions in Section \ref{sec:conclusion}. Proofs of lemmas and theorems are given in supplementary materials. The data and code used in this paper are publicly available ($\tt \href{https://github.com/UncertaintyQuantification/GOLF}{https://github.com/UncertaintyQuantification/GOLF}$).

	\section{Gaussian orthogonal latent factor processes}
	\label{sec:GOLF}

	\subsection{Orthogonal decomposition and posterior independence}
	\label{subsec:orthogonal_independence}
	Let $\mathbf y_s(\mathbf x)=(y_{s_{1}}(\mathbf x),...,y_{s_{n_1}}(\mathbf x))^T$ be an $n_1\times 1$ vector of observations at coordinates $\mathbf s=(\mathbf s_1,...,\mathbf s_{n_1})^T$ with $\mathbf s_i \in \mathbb R^{p_1}$ for $i=1,...,n_1$ and input $\mathbf x \in \mathbb R^{p_2}$. For {{spatially}} correlated data, for instance, $s$ and $x$ denote the latitude and longitude, respectively, and in spatio-temporal models, the spatial coordinates and time {{points}} can be defined as $\mathbf s$ and $x$, respectively.

 Consider the latent factor model: 
	\begin{equation}
	\mathbf y_s(\mathbf x)=\mathbf m_s(\mathbf x)+\mathbf A_s \mathbf z(\mathbf x)+\bm \epsilon, 
	\label{equ:GOLF_model}
	\end{equation}
	where $\mathbf A_s=[\mathbf a_1,...,\mathbf a_d]$ is a ${n_1}\times d$ factor loading matrix and  $\mathbf z(\mathbf x)=( z_1(\mathbf x),..., z_d(\mathbf x))^T$ is a d-dimensional factor processes  with $d\leq {n_1}$,    $ \bm \epsilon \sim \mathcal N(0, \sigma^2_0 \mathbf I_{n_1} )$ being independent Gaussian noises. The mean function  $\mathbf m_s(\mathbf x)=(m_{s_1}(\mathbf x),...,m_{s_{n_1}}(\mathbf x) )^T$ is  typically modeled via a linear trend of regressors, which will be discussed in Section \ref{subsec:mean_function}. 
	  	
		
       	As   data are typically positively correlated at two nearby inputs,  we assume $ z_l(\cdot)$ independently follows a zero-mean Gaussian process (GP), meaning that for any $\{\mathbf x_1,...,\mathbf x_{n_2}\}$,  $\mathbf Z^T_l=(Z_l(\mathbf x_1),...,Z_l(\mathbf x_{n_2}))^T$ is a multivariate normal distribution:  
       	\begin{equation}
       	    (\mathbf Z^T_l  \mid \bm \Sigma_l) \sim \mathcal N(\mathbf 0, \bm \Sigma_l )
       	    \label{equ:Z}
       	\end{equation}
       	 where the $(i,j)$th entry of the covariance matrix is $\sigma^2_lK_{l}(
       	 \mathbf x_i, \mathbf x_j)$ with kernel function $K_l(\cdot,\cdot)$ and variance parameter $\sigma^2_l$, for $l=1,...,d$. Here we assume  independence between the factor processes {\textit{a priori}}. A detailed comparison between our approach and other related approaches is discussed in Section \ref{sec:comparison}.

      	     Note that only the d-dimensional linear subspace of factor loadings $\mathbf A_s$ can be identified  if not further specification of factor loading matrix $\mathbf A_s$ is made, as  the model (\ref{equ:GOLF_model}) is unchanged if the pair $(\mathbf A_s, \mathbf z(\mathbf x))$ is replaced by $(\mathbf A_s \mathbf G, \mathbf G^{-1} \mathbf z(\mathbf x))$  for any invertible matrix $\mathbf G$.  Besides, {{the}} computation could be challenging when the number of factors or input parameters is  large. Thus, we assume that the column of $\mathbf A_s$ is orthonormal.
	
	


	 \begin{assumption}
	 \begin{equation}
	 \mathbf A_s^T \mathbf A_s=\mathbf I_d.
	 \label{equ:A}
	 \end{equation}
	 \label{assumption:A}
	 \end{assumption}
	  Assumption (\ref{assumption:A}) may be replaced by $\mathbf A_s^T  \mathbf A_s=\bm \Lambda$, where $\bm \Lambda $ is a diagonal matrix. {Since we  estimate  variance parameters $\bm \sigma^2=(\sigma^2_1,..., \sigma^2_d)^T$ of  latent factor processes by data,  diagonal terms of  $\bm \Lambda $ are redundant.  Thus we proceed with the Assumption \ref{assumption:A}}.


	


%
%
%
	
	

      Let us first assume we have an $n_1 \times n_2$ matrix of observations $\mathbf Y=[\mathbf y_s(\mathbf x_1),...,\mathbf y_s(\mathbf x_{n_2})]$  at inputs $\{\mathbf x_1,...,\mathbf x_{n_2}\}$, and then we extend our method  to incomplete matrix observations in the  Section \ref{sec:post_sample}. Denote $\mathbf B$ the regression parameters  in the $n_1\times n_2$ mean matrix $\mathbf M=(\mathbf m_s(\mathbf x_1),...,\mathbf m_s(\mathbf x_{n_2}))$. Denote   $\bm \Theta=(\mathbf A_s, \mathbf B, \bm \sigma^2, \bm \gamma)$, which contains the factor loadings, mean parameters, variance parameters and range parameters in the kernel functions. Further let  $\mathbf A_F=[\mathbf A_s,\mathbf A_c]=[\mathbf a_1,\mathbf a_2,...,\mathbf a_{n_1}]$, where $\mathbf A_c$ is an $n_1\times (n_1-d)$ matrix of the orthogonal complement of $\mathbf A_s$.  Assumption \ref{assumption:A} allows us to decompose the marginal likelihood (after integrating out the random factor  $\mathbf Z$) into a product of multivariate normal densities of the outcomes at the projected coordinates:
	\begin{equation}
	p(\mathbf Y \mid \bm \Theta)= \prod^{d}_{l=1} \mathcal{PN}(\tilde {\mathbf y}_l; \mathbf 0, \bm {\tilde \Sigma}_l ) \prod^{n_1}_{l=d+1} \mathcal{PN}(\tilde {\mathbf y}_{l}; \mathbf 0, \sigma^2_0 \mathbf I_{n_1}),  
	\label{equ:marginal_lik}
	\end{equation}
	where  $\tilde {\mathbf y}_l= (\mathbf Y-\mathbf M)^T \mathbf a_l $ for $l=1,...,d$, and $\tilde {\mathbf y}_{l} = (\mathbf Y-\mathbf M)^T  \mathbf a_{l} $ with $\mathbf a_{l}$ being the $(l-d)$th column of $\mathbf A_c$ for $l=d+1,...,n_1$, $\bm {\tilde \Sigma}_l=\bm { \Sigma}_l+\sigma^2_0 \mathbf{I}_{n_2}$  and  $\mathcal{PN}(\cdot\,;\bm \mu,\bm \Sigma)$ denotes the density of the multivariate normal distribution with mean $\bm \mu$ and covariance matrix $\bm \Sigma$. In practice, note that we can avoid computing $\mathbf A_c$ by using the identity $\mathbf A_s \mathbf A^T_s+\mathbf A_c \mathbf A^T_c=\mathbf I_{n_1}$.
 The derivation of Equation (\ref{equ:marginal_lik})    is given in the supplementary materials.

	 The orthogonal factor loading matrix in Assumption \ref{assumption:A} and prior independence of factor processes lead to  the posterior independence of the factor processes, introduced in the following corollary.	
	

	
	%
	
	
		\begin{corollary}
	 For model (\ref{equ:GOLF_model}) with Assumption \ref{assumption:A}:
	\begin{enumerate}
	\item The covariance of the posterior marginal distributions of any two factor processes is zero: $\Cov[\mathbf Z^T_l, \mathbf Z^T_m\mid \mathbf Y, \bm \Theta]=\mathbf 0_{n_2\times n_2}$, where  $l=1,...,d$, $m=1,...,d$ and $l\neq m$.
	
	\item For $l=1,...,d$, the posterior distribution  $(\mathbf Z^T_l \mid \mathbf Y, \bm \Theta)$ follows a multivariate normal distribution
		\begin{equation}
		 \mathbf Z^T_l \mid \mathbf Y, \bm \Theta \sim \mathcal N\left( \bm { \mu}_{Z_l}, \,  { \bm \Sigma}_{Z_l}  \right), 
		 \label{equ:Z_t_l_post}
		 \end{equation}
	where $  \bm { \mu}_{Z_l}=  \bm \Sigma_l \bm {\tilde \Sigma}_l^{-1}  \mathbf {\tilde y}_l$ and ${ \bm \Sigma}_{Z_l}  =\bm \Sigma_l-\bm \Sigma_l \bm {\tilde \Sigma}_l^{-1}\bm \Sigma_l $ with $\bm {\tilde \Sigma}_l=\bm \Sigma_l +\sigma^2_0 \mathbf I_{n_2}$.

	\end{enumerate}
		\label{cor:ind_post}
	\end{corollary}
	We call the latent factor processes in (\ref{equ:GOLF_model}) with Assumption \ref{assumption:A}  \textit{Gaussian orthogonal latent factor} (GOLF) processes, because of  orthogonal decomposition of the likelihood function and  posterior independence between two factor processes.  {The main idea  is to decompose the likelihood of GP models with multi-dimensional inputs by a product of densities with low dimension input and to utilize the continuous-time Kalman filter for fast computation}. As we will see in Section \ref{sec:post_sample}, these two properties dramatically ease the computational burden.

\subsection{Flexible mean function and marginalization}
\label{subsec:mean_function}



The mean function $m_s(\cdot)$ plays an important role in modeling and predicting correlated data.   Computer models (such as {the} numerical solution of partial differential equations), for example, can be included as a part of the mean in an inverse problem \citep{kennedy2001bayesian}. 
Here for simplicity, we use only a linear basis function of $\mathbf s$ and $\mathbf x$, whereas additional terms may be included in the mean if available.

In a GP model, the regression coefficients are often assumed to  be the same across one basis function. 
For instance, the mean function may be modeled as $\mathbf m_{s}(\mathbf x)=\mathbf h_{1}(\mathbf s) \mathbf b_{1,0}$, or $\mathbf m_{s}(\mathbf x)=\mathbf h_{2}(\mathbf x) \mathbf b_{2,0}$, where $\mathbf h_{1}(\mathbf s)$ and   $\mathbf h_{2}(\mathbf x)$ are a set of $1\times q_1$ and $1\times q_2$ mean basis functions with $\mathbf b_{1,0}$ and $\mathbf b_{2,0}$ being $q_1\times 1$ and $q_2\times 1$ regression coefficients, respectively. The regression coefficients $\mathbf b_{1,0}$,  for example,  are shared across each $\mathbf x$.

\begin{figure}[t]
		\begin{tabular}{cc}
		 \vspace{-.4in}
		\includegraphics[height=.3\textwidth,width=.45\textwidth ]{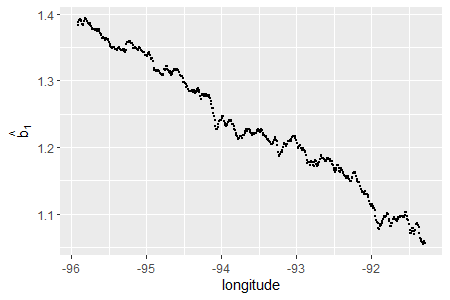} \hspace{.1in}
		\includegraphics[height=.3\textwidth,width=.45\textwidth ]{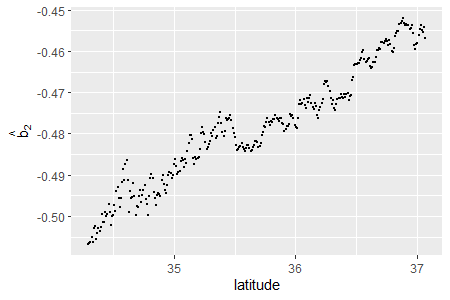}
		\vspace{.25in}

		\end{tabular}
	\caption{Estimated linear coefficients for temperature observations in \cite{heaton2019case}. In the left panel, the dots are the estimated coefficients in a linear regression of observations at each longitude separately using  latitudes as  regressors. The estimated linear coefficients for the observations at each latitude are graphed in the right panel, where  longitudes are used as  regressors.}
	\label{fig:Temp_obs_EDA}
\end{figure}



 The shared regression coefficients may  be a restrictive assumption when data sets are large.  Consider, for instance, the temperature data set used in \cite{heaton2019case}, where the temperature values are shown in Figure \ref{fig:heaton_pred}. In Figure \ref{fig:Temp_obs_EDA}, we graph the fitted linear regression coefficients using latitudes or longitudes as regressors. The estimated regression coefficients are not the same across latitude or longitude.   
 A natural extension of modeling the mean function, therefore, is to allow the mean parameters at each row or column of the observations to be different, e.g. $\mathbf m_{s_i}(\mathbf x_j)=\mathbf h_{1}(\mathbf s_i) \mathbf b_{1,j}$, or $\mathbf m_{s_i}(\mathbf x_j)=\mathbf h_{2}(\mathbf x_j) \mathbf b_{2,i}$, for $i=1,...,n_1$ and $j=1,...,n_2$. Some choices of the  individual mean functions are summarized in Table \ref{tab:mean_function}. 
The mean function may be specified based on model interpretation or exploratory data analysis. Models with different regression coefficients across different types of coordinates are  more suitable to model a large number of observations, as they are more flexible to capture the trend. 
 \begin{table}[t]
\begin{center}
\begin{tabular}{lcccc}
  \hline
    Individual mean & $\mathbf m_{s_i}(\mathbf x_j)$ & $\mathbf M$ & coefficients $\mathbf B$\\
   \hline
  Linear trend of $\mathbf s$   &$\mathbf h_{1}(\mathbf s_i)  \mathbf b_{1,j}$ & $\mathbf H_1 \mathbf B_1$ & $\mathbf B_1$ \\
    Linear trend of $\mathbf x$ &$\mathbf h_{2}(\mathbf x_j)  \mathbf b_{2,i}$  &$(\mathbf H_2 \mathbf B_2)^T$ & $\mathbf B_2$  \\
   Mixed linear trend  &$\mathbf h_{1}(\mathbf s_i)  \mathbf b_{1,j}+\mathbf h_{2}(\mathbf x_j)  \mathbf b_{2,i}$& $\mathbf H_1 \mathbf B_1+(\mathbf H_2 \mathbf B_2)^T$ & $[\mathbf B_1,\mathbf B_2]$ \\
%
%
\hline
\end{tabular}
\end{center}
   \caption{Summary of the mean function studied in this work. In the third column, $\mathbf H_1=(\mathbf h^T_{1}(\mathbf s_1),...,\mathbf h^T_{1}(\mathbf s_{n_1}))^T$ and  $\mathbf H_2=(\mathbf h^T_{2}(\mathbf x_1),...,\mathbf h^T_{2}(\mathbf x_{n_2}))^T$ are $n_1\times q_1$ and $n_2\times q_2$ mean basis matrices, respectively. Regression coefficients are denoted as $\mathbf B_1=(\mathbf b_{1,1}, ...., \mathbf b_{1,n_2})$ and $\mathbf B_2=(\mathbf b_{2,1}, ...., \mathbf b_{2,n_1})$ for the basis function $\mathbf h_1(\cdot)$ and $\mathbf h_2(\cdot)$, respectively.   }

   \label{tab:mean_function}
\end{table}

 To implement full Bayesian inference of the parameters, one may sample from the  posterior distribution of regression parameters $p(\mathbf B\mid \bm \Theta_{-B}, \mathbf Y, \mathbf Z)$. However, we found a severe  identifiability problem between the mean $\mathbf M$  and $\mathbf A \mathbf Z$, when the regression coefficients $\mathbf B$  are sampled from the full posterior distribution.  This is because the likelihood function of the mean parameters is flat when data are very correlated. Consequently, the absolute values of the entries of these two matrices can be both big, making the MCMC algorithm  very unstable. To alleviate the identifiability problem, we first integrate out factors  and sample regression parameters from the marginal posterior distribution $p(\mathbf B\mid \bm \Theta_{-B}, \mathbf Y)$. The marginal posterior distributions of the regression parameters are given in the following Theorem \ref{thm:marginal_post_B_1_B_2} and Theorem \ref{thm:sample_B12}.





\begin{theorem}
\label{thm:marginal_post_B_1_B_2}
\begin{enumerate}
\item (Row regression coefficients). Assume  $\mathbf M=\mathbf H_1 \mathbf B_1$ and the objective prior $\pi(\mathbf B_1)\propto 1$ for $\mathbf B_1$. After marginalizing out the factor $\mathbf Z$, the posterior samples of $\mathbf {  B}_1$  from $p(\mathbf {  B}_1 \mid \mathbf Y, \bm \Theta_{-B_1})$  can be obtained by 
\begin{equation}
\mathbf B_1= \mathbf {\hat B}_1 +  (\mathbf H^T_1 \mathbf H_1)^{-1} \mathbf H^T_1 \mathbf A_s \mathbf {\tilde B}^T_{1,0,s}+ \sigma_0 (\mathbf H^T_1 \mathbf H_1)^{-1} \mathbf H^T_1 (\mathbf I_{n_1}- \mathbf A_s \mathbf A^{T}_s ) \mathbf Z_{0,1}
\label{equ:sample_B_1}
\end{equation}
where $\mathbf {\hat B}_1= (\mathbf H^T_1 \mathbf H_1)^{-1} \mathbf H^T_1 \mathbf Y$, $ \mathbf {\tilde B}_{1,0,s}$ is an $n_2 \times d$ matrix with the $lth$ column independently sampled from $\mathcal N(\mathbf 0, \bm {\tilde \Sigma}_l)$ for $l=1,...,d$, and $\mathbf Z_{0,1}$ is an $n_1 \times n_2$ matrix with each entry independently sampled from the standard normal distribution.

\item (Column regression coefficients). 
Assume $\mathbf M=(\mathbf H_2 \mathbf B_2)^T$ and the objective prior $\pi(\mathbf B_2)\propto 1$ for the regression parameters $\mathbf B_2$. After marginalizing out the factor $\mathbf Z$, the posterior samples of $\mathbf {  B}_2$  from $p(\mathbf {  B}_2 \mid \mathbf Y, \bm \Theta_{-B_2})$ can be obtained by	
	\begin{equation}
\mathbf B_2= \mathbf {\hat B}_2 +    \mathbf {\tilde B}_{2,0,s} \mathbf A^T_s+ \sigma_0 \mathbf L_{H_2} \mathbf Z_{0,2} (\mathbf I_{n_1}- \mathbf A_s \mathbf A^{T}_s ), 
\label{equ:sample_B_2}
\end{equation}
where $\mathbf {\hat B}_2= \sum^d_{l=1} (\mathbf H^T_2 \bm{\tilde \Sigma}_l^{-1} \mathbf H_2)^{-1} \mathbf H^T_2 \bm{\tilde \Sigma}^{-1}_l \mathbf Y^T\mathbf a_l \mathbf a^T_l + (\mathbf H^T_2 \mathbf H_2)^{-1} \mathbf H^T_2 \mathbf Y^T (\mathbf I_{n_1}-\mathbf A_s \mathbf A^T_s)  $ and  $\mathbf {\tilde B}_{2,0,s}$ is a $q_2\times d$ matrix with the $l$th column independently sampled from $\mathcal N(\mathbf 0, (\mathbf H^T_2 \bm {\tilde \Sigma}^{-1}_l  \mathbf H_2)^{-1})$ for $l=1,...,d$. $\mathbf L_{H_2}$ is a $q_2 \times q_2$ matrix such that  $\mathbf L_{H_2}\mathbf L^T_{H_2}=(\mathbf H^T_2 \mathbf H_2)^{-1}$ and   $\mathbf Z_{0,2}$ is a $q_2 \times n_1$ matrix with each entry independently sampled from the standard normal distribution.

	
\end{enumerate}

\end{theorem}

When both the row regression coefficients and column regression coefficients are in the model, we found that  $\mathbf M_1=\mathbf H_1 \mathbf B_1$ and  $\mathbf M_2=(\mathbf H_2 \mathbf B_2)^T$ are not identifiable,  if we sample $\mathbf B_1$ and $\mathbf B_2$ from the full conditional distribution. To avoid this problem, we first marginalizing out $\mathbf B_2$ and  $\mathbf Z$ to sample $\mathbf B_1$ and then we condition $\mathbf B_1$ to sample $\mathbf B_2$. 

\begin{theorem}
\label{thm:both_mean}
Assume $\mathbf M=\mathbf H_1 \mathbf B_1 + (\mathbf H_2 \mathbf B_2)^T$ and let the objective prior $\pi(\mathbf B_1, \mathbf B_2)\propto 1$ for the regression parameters $\mathbf B_1$ and $\mathbf B_2$. 
\begin{enumerate}
\item After marginalizing out $\mathbf Z$ and $\mathbf B_2$, the marginal posterior sample of $\mathbf B_1$ from 	$p(\mathbf {  B}_1 \mid \mathbf Y, \bm \Theta_{-B_1,-B_2})$ can be obtained by 
\begin{equation}
\mathbf B_1= \mathbf {\hat B}_1 +  (\mathbf H^T_1 \mathbf H_1)^{-1} \mathbf H^T_1 \mathbf A_s \mathbf {\tilde B}^T_{1,Q}+ \sigma_0 (\mathbf H^T_1 \mathbf H_1)^{-1} \mathbf H^T_1 (\mathbf I_{n_1}- \mathbf A_s \mathbf A^{T}_s ) \mathbf Z_{0,1} \mathbf P_0,
\label{equ:sample_B_1_both}
\end{equation}
where $\mathbf {\hat B}_1= (\mathbf H^T_1 \mathbf H_1)^{-1} \mathbf H^T_1 \mathbf Y$, $ \mathbf {\tilde B}_{1,Q}$ is an $n_2 \times d$ matrix with the $lth$ column independently sampled from $\mathcal N(\mathbf 0, \mathbf Q_{1,l})$, with $\mathbf Q_{1,l}=\mathbf P_l \bm {\tilde \Sigma}^{-1}_l \mathbf P_l$ where $\mathbf P_l= \mathbf I_{n_2} - \mathbf{H}_{2}(\mathbf{H}^T_{2} \bm {\tilde \Sigma}_l^{-1} \mathbf{H}_{2})^{-1}\mathbf{H}^T_{2}  \bm {\tilde \Sigma}_l^{-1}$ for $l=1,...,d$. $\mathbf Z_{0,1}$ is an $n_1 \times n_2$ matrix with each entry independently sampled from standard normal distribution and $\mathbf P_0=(\mathbf{I}_{n_2} - \mathbf{H}_2(\mathbf{H}^T_2\mathbf{H}_2)^{-1}\mathbf{H}^T_2)$. 

\item  Posterior samples of $\mathbf {  B}_2$ from $p(\mathbf {  B}_2 \mid \mathbf Y_{B_1}, \bm \Theta_{-B_2})$ can be obtained through equation (\ref{equ:sample_B_2}) by replacing $\mathbf Y$ by $ \mathbf Y- \mathbf H_1 \mathbf B_1$.

\end{enumerate}

\label{thm:sample_B12} 
\end{theorem}

In Theorem \ref{thm:marginal_post_B_1_B_2} and Theorem \ref{thm:sample_B12}, the marginal posterior distribution of the regression coefficients  depends on  the $n_1\times d$ factor loading matrix, but not the complement of the factor loading matrix ($\mathbf A_c$). Since we do not need to  compute $\mathbf A_c$,  the most computationally intensive terms are those containing the  covariance matrix $\bm \Sigma_l$ and its inverse. Fortunately, each term can be computed with linear complexity with respect to $n_2$ instead of $n_2^3$ when the Mat{\'e}rn covariance is used, discussed in Section \ref{subsec:KF}.




\subsection{Spatial latent factor loading matrix}


This section discusses a model of the latent factor loading matrix $\mathbf A_s$ that satisfies the orthogonal constraint in (\ref{equ:A}).  As  output values are marginally correlated at  two inputs $\mathbf s_a$ and $\mathbf s_b$, 
a natural choice is to  let $\mathbf A_s$ be the eigenvectors corresponding to the largest $d$ eigenvalues in the eigendecomposition  of the correlation matrix $\mathbf R_s$,  where the $(i,j)$th entry is specified by a kernel function $K_s(\mathbf s_i, \mathbf s_j)$, for $1\leq i,j\leq n_1$. We give a few examples of models that can be written as special cases of the GOLF model when the $\mathbf A_s$ is specified as eigenvectors of $\mathbf R_s$. For simplicity, we assume the mean is zero.  The first and second classes of models are the GP models with  separable covariance functions of input with two dimensions and three dimensions, respectively.	
	
	\begin{example}[Spatial model with separable covariance]
	\label{eg:separable_spatial_model}
	Consider a spatial model of $\mathbf Y$ at  a regular $n_1 \times n_2$ lattice, where the $(i,j)$th input is $(s_{i},x_{j})$ with $s_i$ and $x_j$ denoting the $i$th latitude coordinate and $j$th longitude coordinate, respectively. Assume the covariance of the spatial process is separable, meaning that $\mathbf Y \sim \mathcal N(\mathbf 0, \sigma^2 \mathbf R_s \otimes \mathbf R_x +\sigma^2_0\mathbf I_{n_1n_2})$, where the  $(l_1,m_1)$ term of $\mathbf R_s$ is parameterized by the kernel function $K_s(s_{l_1},s_{m_1} )$ and the $(l_2,m_2)$ term of $\mathbf R_x$ is $K_x(x_{l_2},x_{m_2} )$ for $1\leq l_1,m_1\leq n_1$ and $1\leq l_2,m_2\leq n_2$. Let $\mathbf R_s=\mathbf U_s \bm \Lambda_s \mathbf U_s^T$, where $\mathbf U_s$   is a matrix of eigenvectors and  $\bm \Lambda_s$ is a diagonal matrix of eigenvalues of $\mathbf R_s$ with the $l$th diagonal term  $ \lambda_{l}$.
	  The density  of this spatial model is equivalent to model (\ref{equ:GOLF_model}) with $\mathbf A_s=\mathbf U_s$, $\bm \Sigma_l=\sigma^2 \lambda_{l} \mathbf R_x$ and $d=n_1$. 
	\end{example}

	\begin{example}[Spatio-temporal model with separable covariance]
	\label{eg:separable_spatio_temporal_model}
Consider a spatio-temporal model of $\mathbf Y$ at $n_{1,1} \times n_{1,2} \times n_2$ lattice, where the $(i,j,k)$th input is $(s_{1,i},s_{2,j},x_{k} )$, with $s_{1,i}$ and $s_{2,j}$ denoting the $i$th latitude coordinate and $j$th longitude coordinate, respectively, and $x_k$ denoting the $k$th time point. Let $n_1=n_{1,1}\times n_{1,2}$. Assume the covariance of the spatio-temporal process is separable, meaning that $\mathbf Y \sim \mathcal N(0, \sigma^2 \mathbf R_{s_1} \otimes \mathbf R_{s_2}\otimes \mathbf R_x +\sigma^2_0\mathbf I_{n_1\times n_2})$ with the $(l_i, m_i)$th term of $\mathbf R_{s_i}$ parameterized by the kernel function $K_s(s_{l_i}, s_{m_i})$ with $1\leq l_i,m_i\leq n_{1,i}$ for $i=1,2$, and the  $(l_3, m_3)$th term of $\mathbf R_{x}$ being $K_x(x_{l_3}, s_{m_3})$ with $1\leq l_3,m_3\leq n_2$. Let $\mathbf R_{s_i}=\mathbf U_i \bm \Lambda_i \mathbf U^T_i$ where $\mathbf U_i$  is a matrix of eigenvectors and  $\bm \Lambda_i$ is a diagonal matrix of eigenvalues $\lambda_{l_i}$ for $1\leq l_i\leq n_{1,i}$ and $i=1,2$. 	  The density  of this  spatio-temporal model is equivalent  to model (\ref{equ:GOLF_model}) with $\mathbf A_s=\mathbf U_1 \otimes \mathbf U_2$, $\bm \Sigma_l=\sigma^2 \lambda_{l_1} \lambda_{l_2} \mathbf R_x$ with $1\leq l_i,m_i\leq n_{1,i}$ for $i=1,2$, $l=l_1+(l_2-1)n_{1,2}$ and $d=n_1$.

	\end{example}



 The separable covariance is widely used in emulating and calibrating computationally expensive computer models with scalar output \citep{sacks1989design} and vector output \citep{conti2010bayesian,paulo2012calibration}, whereas the isotropic covariance, {{i.e.,}} the covariance as a function of Euclidean distance of inputs, is used more often in modeling spatially correlated data \citep{gelfand2010handbook}. Some  anisotropic kernels, such as the geometrically {{anisotropic}} kernel, were studied in \cite{zimmerman1993another} for modeling spatially correlated observations.  Note that  the covariance of GOLF processes  in (\ref{equ:GOLF_model}) is not separable in general, as the variance and kernel parameters of each factor process $z_l(\cdot)$ can be different. Different {{kernel}} parameters make the model more flexible, as the factor processes corresponding to large eigenvalues are often found to be smoother than the ones corresponding to small eigenvalues. Separable covariance may be restrictive in this regard as  factor processes are assumed to have the same kernel and parameters. 
 
 

 



Computing the likelihood of GP with separable covariance on a complete  $n\times n$ lattice data generally takes $O(N^{3/2})$ operations through eigen-decomposition of sub covariance matrices. This work generalizes this approach to nonseparable covariance for  both complete and incomplete lattice observations.  {One can further reduce the computational complexity by selecting $d$ eigenvectors corresponding to the $d$ largest eigenvectors from the eigendecomposition of the correlation matrix $\mathbf R_s$. The proportion of summation of the $d$ largest eigenvalues over the summation of total eigenvalues  shall be chosen as large as possible to allow the model to explain the most variability of the signal \citep{higdon2008computer}. We  found that using more factors than the truth typically will not incur {{a}} large reduction of predictive accuracy, whereas using {{a}} much smaller number of factors than the truth will cause a large predictive error  (Example \ref{eg:compare_diff_factors} in simulated studies). Thus one should be cautious {{about}} using {{a}} very small number of factors.}



	\subsection{Kernel functions}
	\label{subsec:kernel}
			We first discuss the kernel function for the factor process $Z_l(\cdot)$, $l=1,...,d$. 
			We assume a product kernel between the inputs \citep{sacks1989design}, i.e. for any input $\mathbf x_a=(x_{a1},...,x_{ap_2})$ and $\mathbf x_b=(x_{b1},...,x_{bp_2})$,  $K_{l}(\mathbf x_a, \mathbf x_b  )=\prod^{p_2}_{i=1}K_{l,i}(|x_{ai}-x_{bi}|)$, where $K_{l,i}(\cdot)$ is a kernel of the $l$th coordinate of the input for $l=1,...,d$ and $i=1,...,p_2$. 
			
			We focus on Mat{\'e}rn covariance \citep{handcock1993bayesian} as kernel function $K_{l,i}(\cdot)$ in this work. Each kernel contains   positive roughness parameter $\nu_{l,i}$ and a nonnegative range parameter $\gamma_{l,i}$ for $l=1,...,d$ and $i=1,...,p_2$.  
The roughness parameter of the Mat{\'e}rn kernel controls the smoothness of the process. When  $\nu_{l,i}=\frac{1}{2}$, the Mat{\'e}rn kernel becomes the exponential kernel:  $K_{l,i}(|x_{ai}-x_{bi}|)=\exp(-|x_{ai}-x_{bi}|/\gamma_{l,i})$, and when $\nu_{l,i}\to \infty$, the  Mat{\'e}rn kernel becomes  the Gaussian kernel:  $K_{l,i}(|x_{ai}-x_{bi}|)=\exp(-|x_{ai}-x_{bi}|^2/(2\gamma^2_{l,i}))$. The half-integer Mat{\'e}rn kernel (i.e. $(2\nu_{l,i}+1)/2 \in \mathbb N$) has a closed form expression.
	When  $\nu_{l,i}=5/2$, for example, the Mat{\'e}rn kernel is
 \begin{equation}
K_{l,i}(|x_{ai}-x_{bi}|)=\left(1+\frac{\sqrt{5}|x_{ai}-x_{bi}|}{\gamma_{l,i}}+\frac{5|x_{ai}-x_{bi}|^2}{3\gamma_{l,i}^2}\right)\exp\left(-\frac{\sqrt{5}|x_{ai}-x_{bi}|}{\gamma_{l,i}}\right) \,,
\label{equ:matern_5_2}
\end{equation}
for $l=1,...,d$ and $i=1,...,p_2$. 

	In constructing GOLF processes, we decompose the density of the GP model with multi-dimensional input into a product of the orthogonal components with {{lower-dimensional}} input. This is because the likelihood and the predictive distribution of a GP model with a half-integer Mat{\'e}rn covariance can be computed through linear operations with respect to the sample size by the continuous-time Kalman filter \citep{sarkka2012infinite} when $p_2=1$. 
	The computational advantage will be discussed in Section \ref{subsec:KF}. 
	
	
	For the factor loading matrix, we let $\mathbf A_s$ be the first $d$ eigenvectors of $\mathbf R_s$. The kernel functions for $\mathbf R_s$ can be chosen similarly as the kernel for the latent factor processes. Without the loss of generality, we assume $\mathbf R_s$ is parameterized by a product kernel  with the range parameters $\bm \gamma_0$, and the Mat{\'e}rn kernel being used for each coordinate of $\mathbf s$.

	\section{Posterior sampling for GOLF processes}
	\label{sec:post_sample}
	
	\subsection{A Markov chain Monte Carlo approach}
	\label{subsec:MCMC}
	In many applications, the observations contain missing values. Denote  $\mathbf Y^o_v$ and $\mathbf Y^u_v$ the vectors of  observed data and missing data in matrix $\mathbf Y$ with size $N_o$ and $N_u$, respectively. 
	Directly computing the likelihood includes calculating the inverse and determinant of an $N_o\times N_o$ covariance matrix, which has computational operations $O(N^3_o)$ in general, making it infeasible for large  number of observations. 	Here we discuss a computationally feasible way for the GOLF model when observations are from  incomplete matrices. 
	
		
		
	    We start with a set of initial values at the locations with missing observations. Denote  $\mathbf Y^{(t)}_v=\mbox{vec}(\mathbf Y^{(t)})=[(\mathbf Y^o_v)^T, (\mathbf Y^{u,{(t)}}_v)^T]^T$  an $N$-vector, where $\mathbf Y^o_v$   and $\mathbf Y^{u,{(t)}}_v$ are vectors of observations and samples at the missing locations in the  $t$th iteration, $t=1,...,T$. First, we use a Metropolis algorithm to sample  $\bm \Theta^{(t+1)}$ from the marginal posterior distribution $p(\bm \Theta \mid \mathbf Y^{(t)})$, where the marginal density  is given in Equation (\ref{equ:marginal_lik}). 
		In the second step, we sample $\mathbf Z^{(t+1)}_l$ from  $p(\mathbf Z_l^{(t+1)} \mid \mathbf Y^{(t)},\bm \Theta^{(t+1)} )$ by Equation (\ref{equ:Z_t_l_post}) for $l=1,...,d$, and then we generate $\mathbf Y^{(t+1)}=\mathbf A^{(t+1)} \mathbf Z^{(t+1)} + \mathbf E^{(t+1)}$, where $\mathbf E^{(t+1)}$ is an $n_1 \times n_2$ matrix of white noise with variance $\sigma^{(t+1)}_0$ and $\mathbf A^{(t+1)}$ is a $n_1\times d$ matrix of the $d$ eigenvectors corresponding to the $d$ largest eigenvalues from the eigendecomposition of the correlation matrix $\mathbf R_s$ in the $(t+1)th$ iteration. We can obtain $\mathbf Y^{u,(t+1)}_v$ by the last $N_u$ terms in $\mathbf Y^{(t+1)}_v$, for $t=1,...,T$. Note that the observed data $\mathbf Y^o_v$  is never changed.


	 
	 For computational reasons, we define the nugget parameter in each kernel (i.e. the inverse of the signal variance to the noise variance ratio  parameter) $\eta_l=\sigma^2_0/\sigma^2_l$ for $l=1,2,...,d$, and the inverse range parameter $\beta_{l,i}=1/\gamma_{l,i}$, where $i=1,...,p_1$ when $l=0$, and $i=1,...,p_2$ when $l\geq 1$. The transformed parameters $\bm {\tilde \Theta}$ contain the mean parameters $\mathbf B$, inverse range parameters $\bm \beta=(\bm \beta_0,...,\bm \beta_d)$, nugget parameters $\bm \eta=(\eta_1,...,\eta_d)$ of the factor processes and the variance of the noise $\sigma^2_0$. 
	 
	 
	 For  mean and noise variance parameters, we use an objective prior $\pi^{R}(\mathbf B, \sigma^2_0)\propto 1/\sigma^2_0$.  We assume the jointly robust (JR) prior for the kernel parameters: $\pi^{JR}(\bm \beta_{l},\eta_l)\propto( \sum^{p_2}_{i=1}(c_{l,2}\beta_{l,i}+\eta_l) )^{c_{l,1}}\exp(-c_{l,3}\sum^{p_2}_{i=1}(c_{l,1}\beta_{l,i}+\eta_l))$  with default parameters $c_{l,1}=1/2-p_2$, $c_{l,2}=1/2$, and $c_{l,3}$ being the average distance between the $l$th coordinate of two inputs for $l=1,...,d$ \citep{gu2018jointly}. Note here  $c_{l,1}=1/2-p_2$ is the default parameter for the MCMC algorithm, whereas this prior parameter is different if one maximizes the marginal posterior distribution. The jointly robust prior is equivalent to the inverse gamma prior when the input dimension is one without a nugget parameter. The inverse gamma prior is assumed for each coordinate of $\bm \beta_0$ with shape and rate parameter being $-1/2$ and $1$, respectively. The JR prior  can alleviate the potential numerical problem when the estimated range and nugget parameters are close to the boundary of the parameter space, as the density of the JR prior is close to zero at these scenarios. As the sample size is large, the bias inserted from the prior is small.

	 
  \begin{algorithm}[t]
\caption{MCMC algorithm when the kernel parameters are different}
\raggedright
(1) For $l=1,...,d$,  sample $(\bm \beta^{(t+1)}_{l}, \eta^{(t+1)}_{l}) $ from $p( \bm \beta_{l}, \eta_{l}\mid  \mathbf {\tilde y}^{(t)}_{l} )$. 

(2) Sample  $\bm \beta^{(t)}_{0}$ from
$p( \bm \beta^{(t)}_{0} \mid  \mathbf {Y}^{(t)}, \bm \beta^{(t+1)}_{1:d}, \bm \eta^{(t+1)}_{1:d}, \mathbf B^{(t)})$. 



(3) Sample   $ \sigma^{(t+1)}_{0} $ from  $p(  \sigma^{(t+1)}_0 \mid  \mathbf {Y}^{(t)}, \bm \beta^{(t+1)},\bm \eta^{(t+1)},\mathbf B^{(t)} )$. 



 (4) Sample $\mathbf B^{(t+1)}$ from $p( \mathbf B^{(t+1)} \mid  \mathbf {Y}^{(t)},\bm \beta^{(t+1)}, \bm \eta^{(t+1)})$. Update the mean matrix $\mathbf M^{(t+1)}$ and the projected observations $\mathbf {\tilde  y}^{(t)}_l=(\mathbf Y-\mathbf M^{(t+1)})^T \mathbf a_l $.

(5) For $l=1,...,d$, sample $\mathbf Z^{(t+1)}_l$ from $p(\mathbf Z_l^{(t+1)} \mid \mathbf {\tilde y}^{(t)}_{l},\bm \beta^{(t+1)}, \bm \eta^{(t+1)} )$ by Corollary \ref{cor:ind_post} and sample $\mathbf Y^{(t+1)}$ by model (\ref{equ:GOLF_model}). Update $\mathbf Y^{u,(t+1)}_v$ by the last $N_u$ terms in $\mathbf Y^{(t+1)}_v$ and let $\tilde {\bm  y}^{(t+1)}_l=(\mathbf Y^{(t+1)}-\mathbf M^{(t+1)})^T \mathbf a_l$.  

(6) Update the posterior  $p( \bm \beta^{(t+1)}_{l}, \eta^{(t+1)}_{l}\mid  \mathbf {\tilde y}^{(t+1)}_{l} )$ and go back to (1) when $t<T$.
\label{algorithm:1_dim}
\end{algorithm}

		 

	 The MCMC algorithm of the GOLF model is given in Algorithm \ref{algorithm:1_dim}. 	In step (1) to step (4) of Algorithm \ref{algorithm:1_dim}, we marginalize out the factor processes to compute the posterior distribution of the parameters. This is critically important as we found severe identifiability problems between the mean matrix $\mathbf M$ and $\mathbf A \mathbf Z$ if the parameters are sampled from the full conditional distributions. Moreover, after marginalizing out the factor processes, the covariance matrix of the distribution  $ \mathcal{PN}(\tilde {\mathbf y}_l; \mathbf 0, \bm {\tilde \Sigma}_l )$  in  (\ref{equ:marginal_lik})  contains a nugget term, which makes the computation stable. 
	

	The Algorithm \ref{algorithm:1_dim} can be easily modified for different scenarios. When the factor processes have the same covariance matrix, we can combine step (1) and step (2) to sample the shared kernel and nugget parameter. Step (4) may be skipped if one has zero-mean or modified if one has the shared regression coefficients in the model.  
	
	
	
	Denote $\bm \Sigma_l= \mathbf L_l  \mathbf L^T_l$ where $\mathbf L_l$ is a lower triangular matrix in the Cholesky decomposition of $\bm \Sigma_l$. We need to efficiently compute  the terms $| \tilde {\bm \Sigma}_l|$,  $\mathbf L^{-1}_l\mathbf v_{l}$,   $\mathbf L_l \mathbf v_{l}$ for any real-valued vector $\mathbf v_{l}:=(v_{l,1},...,v_{l,n_2})^T$ and sample $(\mathbf Z^{(t+1)}_l)^T$ from  $p((\mathbf Z_l^{(t+1)})^T \mid \mathbf {\tilde y}^{(t)}_{l},\bm \beta^{(t+1)}, \bm \eta^{(t+1)} )$ for $l=1,...,d$. Direct computation of the Cholesky decomposition of $\bm \Sigma_l$ requires $O(n^3_2)$ computational operations for each $l=1,...,d$. Luckily, for Mat{\'e}rn covariance with a half-integer roughness parameter and one-dimensional input, computing  any of these terms  only takes $O(n_2)$ operations without approximation.
	
	
	

%
%
%

	
	\subsection{Continuous-time Kalman filter}
	\label{subsec:KF}
	
	We briefly review the continuous-time Kalman filter algorithm  and the connection between the Gaussian Markov random field and GP with Mat{\'e}rn covariance. The spectral density of the Mat{\'e}rn covariance with the half-integer roughness parameter was shown to be the same as a continuous-time autoregressive process defined as a stochastic differential equation (SDE)  \citep{whittle1963stochastic}. Suppose the observations are $\tilde {\mathbf y}_l=(\tilde { y}_{1,1},...,\tilde { y}_{l,n_2})^T$. For $j=1,...,n_2$ and $l=1,...,d$, starting from the initial state $\bm \theta_l(s_{0}) \sim  {\MN}(\mathbf 0, \mathbf W_l(s_{0}) )$, the  solution of the SDE follows \citep{hartikainen2010kalman}: 
    \begin{align}
    \label{equ:ctdlm}
    \begin{split}
    \tilde y_{l,j}&= \mathbf F\bm \theta_l(x_{j}) + \epsilon_{l,j}, \\
    \bm \theta_l(x_{j})&=\mathbf G_l(x_{j-1}) \bm \theta_l(x_{j-1}) +\mathbf w_l(x_j),
    \end{split}
    \end{align}
    where $\mathbf w_l(x_j) \sim \mathcal N(\mathbf 0, \mathbf W_l(s_j))$, $\epsilon_{l,j}$ is an independent white noise for $l=1,...,d$ and $j=1,...,n_2$. For the Mat{\'e}rn kernel with a half-integer roughness parameter, the terms $\mathbf G_l(x_j) $, $\mathbf W_l(x_j)$, and $\mathbf F$ can be expressed explicitly as a function of $|x_j-x_{j-1}|$ and the range parameter of the kernel.  Thus, the forward filtering and backward smoothing algorithm (FFBS) can be applied to compute the likelihood and to make predictions with linear computational  operations of the number of observations (see e.g. Chapter 4 in \cite{West1997} and Chapter 2 in \cite{petris2009dynamic} for the FFBS algorithm). 
    The likelihood function and predictive distribution of a GP model having the Mat{\'e}rn kernel with roughness parameters being $1/2$ and $5/2$ through the FFBS algorithm are implemented in {\sf FastGaSP} package available at CRAN.  The computational complexity of the FFBS algorithm is only $O(n_2)$, with $n_2$ being the number of observations.
    

    We briefly discuss how to apply the FFBS algorithm to compute  terms $\mathbf L^{-1}_l \mathbf { \tilde y}_l$ and $| \tilde {\bm \Sigma}_l|$ needed in Algorithm \ref{algorithm:1_dim},  for $l=1,...,d$. In the FFBS algorithm, the one-step-ahead predictive distribution  $ (\tilde { y}_{l,j} \mid  \tilde { y}_{l,1:{j-1}})\sim \mathcal N(f_{l}(x_j), Q_{l}(x_j))$ can be derived  iteratively for $j=1,...,n_2$ and for each $l=1,...,d$. Closed form expressions of $f_{l}(x_j)$ and $Q_{l}(x_j)$ for the Mat{\'e}rn covariance in (\ref{equ:matern_5_2}) are given in \cite{gu2020fast}. For $l=1,...,d$, we have  following expressions for the computational expensive terms in the likelihood function:
    \[| \tilde {\bm \Sigma}_l|=\prod^{n_2}_{j=1} {Q_l(x_j)}, \quad \mbox{and} \quad \mathbf L^{-1}_l \mathbf { \tilde y}_l=\left(\frac{\tilde y_{l,1}- f_{l,1}}{\sqrt{Q_l(x_1)}},..., \frac{\tilde y_{l,1}- f_{l,n_2}}{\sqrt{Q_l(x_{n_2})}}\right)^T.\]
    
    We use the backward sampling algorithm \citep{petris2009dynamic} to sample $\bm \theta_{l,n_2}$ from $p(\bm \theta_{l,n_2}\mid \mathbf {\tilde y}^{(t)}_{l},\bm \beta^{(t+1)}, \bm \eta^{(t+1)})$ and $\bm \theta_{l,j}$ from $p(\bm \theta_{l,j}\mid \mathbf {\tilde y}^{(t)}_{l},  \theta_{l,j+1},\bm \beta^{(t+1)}, \bm \eta^{(t+1)})$ sequentially, for $j=n_2-1,...,1$. Posterior samples    $\mathbf Z_l^T=\left(\mathbf z_l(x_1),...,\mathbf z_l(x_{n_2})\right)^T$ can be obtained by the first entry of the posterior sample $\bm \theta_{l,j}$ from the backward sampling algorithm, for $j=1,...,n_2$. Furthermore, for any $n_2\times 1$ real vector $\mathbf v_{l}$, we have $\mathbf L_l \mathbf v_{l} =(f_{l,1}+\sqrt{Q_l(x_1)} v_{l,1},...,f_{l,n_2}+\sqrt{Q_l(x_{n_2})}v_{l,n_2})^T$ for $l=1,...,d$ and $j=1,...,n_2$. 
	\subsection{Computational complexity} 
\label{subsec:complexity}
Denote $p=p_1\times p_2$ the total dimension of the inputs $(\mathbf s, \mathbf x)$ and suppose the observational matrix is $n_1\times n_2$ with irregular missing values, where $n_1\leq n_2$ and $N=n_1n_2$. We discuss the computational complexity for three scenarios with $p=2$ (e.g. spatially correlated data), $p=3$ (e.g. spatio-temporal data) and $p>3$ (e.g. functional data). 


When $p=2$, the computational complexity of the GOLF model with the half-integer Mat{\'e}rn kernel is $O(Nd)$. First,  we compute the first $d$ eigenvectors of $\bm \Sigma_s$ to obtain $\mathbf A_s$, which has  $O(n^2_1d)$ operations (see e.g. Chapter 4.5.5 in \cite{Zhaojun2000eigen}). Second, computing the marginal likelihood and sampling the factor processes by the FFBS algorithm only cost  $O(n_2d)$ operations. The largest computational order is from the matrix multiplication $\tilde {\mathbf Y}^T=(\mathbf Y-\mathbf M)^T \mathbf A_s$, which is at the order of  $O(Nd)$. 

For $p=3$, we let $\mathbf A_{s}=\mathbf A_{s_1}\otimes \mathbf A_{s_2}$, where $\mathbf A_{s_1}$ and $\mathbf A_{s_2}$ are the first $d_1$ and $d_2$ eigenvectors of $n_{1,1}\times n_{1,1}$ matrix $\bm \Sigma_{s_1}$ and $n_{1,2}\times n_{1,2}$ matrix  $\bm \Sigma_{s_2}$, respectively,  with $n_{1,1}\times n_{1,2}=n_{1}$ and  $\bm \Sigma_{s_1}\otimes \bm \Sigma_{s_2}=\bm \Sigma_{s}$. Without the loss of generality, assume $d_1\leq d_2$ and $n_1\leq n_2$. Let the total number of factor processes be $d=d_1d_2$. The computational order of the GOLF model with a half-integer Mat{\'e}rn covariance function is $O(n_1n_2d_{max})$ where $d_{max}$ is the maximum of $d_1$ and $d_2$ (noting this is  smaller than  $O(n_1n_2d)$). To see this,  computing the eigendecomposition of $\bm \Sigma_{s_1}$ and $\bm \Sigma_{s_2}$ requires $O(d_1 n^2_{1,1})$ and $O(d_2 n^2_{1,2})$ operations, respectively. Second, using the FFBS algorithm to compute the marginal likelihood and to  sample factor processes  costs  $O(dn_2)$ operations. At last, we do NOT directly compute $\mathbf Y^T\mathbf A_s$ as its computation operations are $O(Nd)$. Instead, we first write the observations as an $n_{2}\times n_{1,2}\times n_{1,1}$  array $\mathbf Y^T_{ar}$, where the $(i,j,k)$th entry being the outcome at $(s_{1,i},s_{2,j},x_{k} )$. Then we  do a 3-mode matrix product followed by a 2-mode matrix product  $\mathbf {\tilde Y}^T_{ar} \times_{3} \mathbf A_{s_1} \times_2 \mathbf A_{s_2} $ \citep{kolda2009tensor}, which has the computation operations $O(n_2 n_1 d_1)$ and $O(n_2 n_{1,2}d)$, respectively.  Finally we concatenate the second and third dimensions of $\mathbf {\tilde Y}^T_{ar}$ to obtain the $n_2\times d$ matrix $\mathbf {\tilde Y}^T$. 


For the case when $p>3$, there might be two scenarios. In the first scenario, the data are observed in an $n_{1,1}\times n_{1,2}\times ...\times n_{1,k} \times n_2 $ tensor  with irregular missing values, where $n_{1,1}\times n_{1,2}\times ...\times n_{1,k}=n_1$. In this scenario, the computation will be $N d_{max}$, where $d_{max}$ is the maximum of $d_1,...,d_k$ with similar deduction for the case with $p=3$. In the second scenario, we have $p_2>1$. Examples include emulating a computationally expensive computer output with multivariate output \citep{conti2010bayesian,paulo2005default}. In this case, the Kalman filter algorithm may not be applied, so the additional computational order is  $O(n^3_2)$, when the covariance of the factor process is the same. If the covariance is not the same, we need to additionally compute the  inverse of covariance matrices of $d$ multivariate normal distributions, which is at the order of $O(dn^3_2)$. 


In sum, the computational complexity of GOLF for all scenarios considered herein is much smaller than $O(N^3_o)$ from directly inverting the covariance matrices. Besides, a few steps in the MCMC algorithm can be computed in parallel, such as FFBS algorithm to compute the product of $d$ marginal densities of projected output and the matrix multiplication $\tilde {\mathbf Y}^T=(\mathbf Y-\mathbf M)^T \mathbf A_s$, to further reduce the computational complexity.

\section{Comparison and connection with other related models}
\label{sec:comparison}

{ GOLF processes are closely connected to a wide range of approaches on approximating  GPs for modeling large correlated data.  Model (\ref{equ:GOLF_model}) is a linear model of coregionalization (LMC) \citep{gelfand2004nonstationary}, where the factor loading matrix is parameterized by input variables. 
Another widely used model for multivariate functional data is the semiparametric latent factor model (SLFM)  \citep{seeger2005semiparametric}, where the factor loading matrix can be estimated by the principal component analysis (PCA) \citep{higdon2008computer}. However, the linear subspace estimated by PCA is equivalent to maximum marginal likelihood estimator (MMLE) with independent factors (\cite{tipping1999probabilistic}), whereas the latent factors at different input variables are assumed to be correlated. The MMLE of factor loadings with correlated factors was derived in \citep{gu2018generalized}, called  the generalized {{probabilistic}} principal component analysis (GPPCA). 
Our approach has two distinctions. First, our approach applies to observations with irregular missing values, whereas  the observations are  required to be matrices in GPPCA. Second, both inputs $\mathbf s$ and  $\mathbf x$ are used for estimation, whereas only the input in latent processes  is used in GPPCA and  predictions can be  more accurate. 
 

 To overcome the computational bottleneck of GPs, we project  observations on orthogonal coordinates in a GOLF model, as the complexity of computing the likelihood of GPs with Mat{\'e}rn covariances with one dimension input is fast by the continuous-time Kalman Filter. 
The  computational complexity can be further reduced   by only using factor processes with large eigenvalues. The reduced rank approach is used widely in modeling  correlated data. For instance, the predictive process by a set  of pre-specified knots was studied in \cite{banerjee2008gaussian},  and the multiresolution local bisqaure functions were used in \cite{cressie2008fixed}. Limitations of {{the}} reduced-rank method are studied in \cite{stein2014limitations}. Note that even for the full rank covariance, the computational order of GOLF is much less than $O(N^3_o)$. The primary goal is not to propose a reduced rank model herein, but to reduce the computational complexity of a GP model with a full-rank, flexible covariance function through orthogonal projections. 



Many other approximation methods for GPs follow the framework of Vecchia's approximation \citep{katzfuss2017general,vecchia1988estimation}. 
Vecchia's approximation is a broad framework that assumes the sparsity of the inverse of Cholesky decomposition of the covariance matrix of the latent processes, where the key is on  selecting the order of the latent variables and imposing sensible conditional independence assumptions between variables.  GOLF processes with Mat{\'e}rn kernel  is closely related to Vecchia's approximation, in the sense that the model can be written as a vector autoregressive model with orthogonal factor loading matrix. Our way of computing  likelihood and predictions based on {{the}} FFBS algorithm is exact, rather than an approximation to the likelihood function. We compare our approach with a few other methods that fall into the  framework of Vecchia's approximation in Section \ref{subsec:real_spatial_data}.  
}






\section{Simulated studies}
\label{sec:simulation}

We discuss two simulated examples in this section. We first study a simulated example with a small sample size to study the predictive performance and  parameter inference between  GOLF processes and  the exact GP model by directly computing the inversion and determinant of the covariance matrix in the likelihood function.
In the second simulated example, we generate  observations from  separable  and nonseparable models to study the predictive performance of GOLF processes with {{a}} different number of factors, and with the same or different kernel parameters. For both examples, we implement $J=100$ experiments in each scenario, and we generate $T=5,000$ MCMC samples for each method with the first $20\%$ of the samples  used as the burn-in samples. 

Denote $y^*_{i,j}$ the $i$th held-out data in the $j$th simulated experiment in each scenario, for $i=1,...,n^*$ and  $j=1,...,J$. Let $\hat y^*_{ij}$ and  $C{I_{ij}}(95\% )$ be the predictive mean and  $95\%$ predictive credible interval of the $i$th held-out data at the $j$th experiment, respectively. For both simulated examples, we record the root mean square error, the percentage of held-out observations percentage  covered in the $95\%$ predictive  interval, and the  average length of the $95\%$ predictive  interval  of the $j$th experiment (${L_{CI_j}(95\%)}$): 
\begin{align}
\text{RMSE}_j&=\sqrt{\frac{\sum^{N^*}_{i=1}(\hat y^*_{ij}-  y^*_{ij})^2 }{N^*}},\, \label{equ:RMSE} \\
P_{CI_j}(95\%) &= \frac{1}{{N^{*}}} {\sum\limits_{i = 1}^{N^{*}} 1\{y^*_{ij}  \in C{I_{ij}}(95\% )\}}\,,\label{equ:PCI} \\
{L_{CI_j}(95\%)} &= \frac{1}{{{N^{*}}}} \sum\limits_{i = 1}^{{N^{*}}} {\Length\{C{I_{ij}}(95\% )\} } \,, \label{equ:LCI}
\end{align}
for $j=1,...,J$.  We compute  average values of these three quantities over $J=100$ simulations to evaluate each approach. A precise method should have {{a}} small average RMSE, $P_{CI}(95\%)$ close to the $95\%$ nominal level, and short predictive interval lengths.  Here we only consider the pairwise interval of responses at each coordinate as outputs are univariate on spatial or spatio-temporal domain. Simultaneous credible interval  can be used for applications with multivariate responses \citep{sorbye2011simultaneous}.


\begin{example}[GOLF processes and exact GP model] Data are sampled from {{a}} zero-mean separable GP model  with two-dimensional inputs at a $25\times 25$ regular lattice in $[0,1]^2$.  Two missing patterns are considered, where the data are missing at random in the first case, and  a disk in the centroid of the lattice is missing in the second case.
\label{eg:compare_exact}
\end{example}



We assume a small sample size in Example \ref{eg:compare_exact} because of the computational burden by the exact Gaussian process model. We use the unit-variance covariance matrix parameterized by the exponential kernel and the Mat{\'e}rn kernel in (\ref{equ:matern_5_2}) to generate the data. The range parameters of Mat{\'e}rn kernel are chosen as $\gamma_0=1$ and $\gamma_1=...=\gamma_d=1/3$.  The range parameters of exponential kernel are chosen to be  $\gamma_0=4$ and $\gamma_1=...=\gamma_d=1$. All the range parameters, the variance of the kernel, and noise are estimated by each method based on the MCMC algorithm. 

We compare   GOLF processes and the exact GP model where the inverse and determinant of the covariance matrix  are directly computed. Both models use the same prior and proposal distribution in the MCMC algorithm to sample the kernel parameters.   Table \ref{tab:simulation_exact_GP}  gives the predictive performance of both methods for three scenarios, where $50\%$ and $20\%$ of the output are missing at random in the first two scenarios, and approximately $20\%$ of the output {{is}} missing in a disk in the centroid of the lattice in the third scenario. Graphs of the observed data, full data,  predictions, and trace plots of the posterior samples  in one simulation are given in the supplementary materials. 






\begin{table}[t]
\centering

\resizebox{\columnwidth}{!}{%
\begin{tabular}{c|cc|ccc|ccc|ccc}
\hline
                                             & \multicolumn{2}{c|}{Missing value} & \multicolumn{3}{c|}{GOLF } & \multicolumn{3}{c|}{Exact GP model} & \multicolumn{3}{c}{Difference} \\ \hline
Kernel                                       & Percentage        & Pattern        & RMSE      & $P_{CI}(95\%)$       & $L_{CI}(95\%)$       & RMSE         & $P_{CI}(95\%)$          & $L_{CI}(95\%)$          & $\Delta$RMSE     & $\Delta$L       & $\Delta$U       \\ \hline
\multirow{3}{*}{Mat{\'e}rn} & 50\%              & random         & 0.106     & 0.954     & 0.425     & 0.106        & 0.952        & 0.423        & 0.002    &0.006   &0.006   \\
                                             & 20\%              & random         & 0.103     & 0.952     & 0.410     & 0.103        & 0.952        & 0.411        & 0.001    & 0.007   & 0.007   \\
                                             & 20\%              & disk           & 0.108     & 0.909     & 0.430     & 0.108        & 0.913        & 0.431        & 0.005    & 0.008   & 0.009   \\ \hline
\multirow{3}{*}{Exp}                         & 50\%              & random         & 0.129     & 0.955     & 0.518     & 0.128        & 0.953        & 0.513        & 0.005   & 0.009	   & 0.008   \\
                                             & 20\%              & random         & 0.120     & 0.947     & 0.472     & 0.120        & 0.948        & 0.471        & 0.003    & 0.009   & 0.009   \\
                                             & 20\%              & disk           & 0.156     & 0.941     & 0.602     & 0.154        & 0.946        & 0.605        & 0.013    & 0.019   & 0.019   \\ \hline
\end{tabular}
}

\caption{\label{tab:simulation_exact_GP} Comparison between the exact GP model and GOLF processes.  $J=100$ simulated experiments are conducted for each scenario. 
$\Delta$RMSE$=\frac{1}{J}\sum^J_{j=1}\Delta$RMSE$_j$ measures the average $L_2$ distance by the two methods, where $\Delta\mbox{RMSE}_j= (\frac{1}{N^*}\sum^{N^*}_{i=1} (\hat y^*_{ij,GOLF}- \hat y^*_{ij,GP})^2)^{1/2}$ with $\hat y^*_{ij,GOLF}$ and $\hat y^*_{ij,GP}$ denote the predictive mean by GOLF processes and exact GP model, respectively. $\Delta$L and $\Delta$U measure the average absolute difference between the lower bound and upper bound of 95\% predictive intervals of the GOLF processes and the exact GP model, respectively. }
\end{table}


As shown in  Table \ref{tab:simulation_exact_GP}, both methods have accurate predictions and uncertainty assessment for all scenarios.  {{Out-of-sample}} RMSE for predicting the held out observations is close to $0.1$, the standard deviation of the  noise. The $95\%$ predictive confidence intervals cover around $95\%$ of the held-out observations, and the average length of the predictive confidence interval is small. Predictions of both methods are more precise for the cases when the data are missing at random than the ones when a disk of output is missing in the centroid of the lattice, as the estimated correlation between the held-out test output and nearby observations are relatively accurate. 


For  Example \ref{eg:compare_exact}, note that  GOLF processes and the exact GP model are the same with two different computational strategies. For GOLF processes, we sample the missing values to use the fast computational strategy, whereas the inverse and determinant of the covariance matrix are computed in the exact GP model directly. Therefore, the two different strategies have significantly different computational operations. The computational operations of GOLF processes is $O(Nd)$ with $N=n_1\times n_2$ ($d=n_1$ in Example \ref{eg:compare_exact}), whereas the computational operations of the exact GP model is $O(N^3_o)$, where $N_o$ is the number of observations.  Thus, GOLF processes are computationally feasible for a large data set. On the other hand, the  difference in predictions and uncertainty assessment between the exact GP model and GOLF is small (last three columns in  Table \ref{tab:simulation_exact_GP}), since  we do not make any approximation in computing GOLF processes.

  \begin{figure}[t]
\centering
  \begin{tabular}{c}
	\includegraphics[height=.33\textwidth,width=1\textwidth]{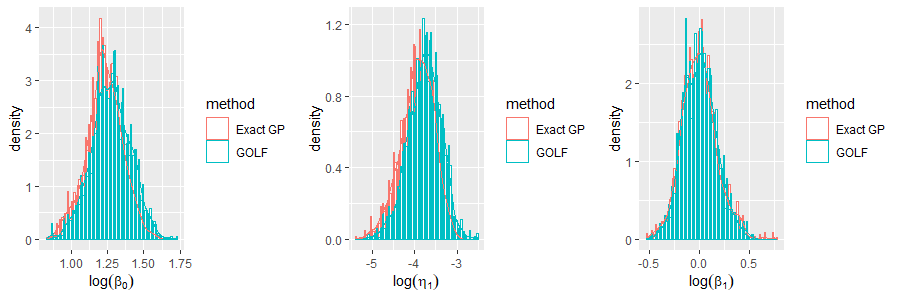}
  \end{tabular}
  \vspace{-.2in}
   \caption{The histogram of posterior samples of the logarithm of the inverse range parameters and nugget parameters in one simulation of Example \ref{eg:compare_exact}, where the data are generated using the Mat{\'e}rn kernel in (\ref{equ:matern_5_2}) with $50\%$ of the values missing at random.}
   
\label{fig:hist_par}
\end{figure}

Figure \ref{fig:hist_par} shows the histogram of the 4000 after burn-in posterior samples from the GOLF processes and exact GP model in one simulation of Example \ref{eg:compare_exact}. The posterior samples of {{the}} two methods are close to each other.  The difference becomes even smaller when we increase the number of MCMC samples.


\begin{example}[GOLF processes with different number of factors and kernel parameters]  The data are sampled from two scenarios with {{two-dimensional}} inputs being a $100\times 100$ lattice in $[0,1]^2$. In the first scenario, the range parameters of the kernel of each factor process are the same, whereas these parameters are chosen to be different  in the second scenario. In both scenarios, a disk of output in the centroid of the lattice is masked out for testing, corresponding to approximately $20\%$ of the total number of data. We use $d=30$ (low-rank)  and {\bf  $d=100$} (full-rank) factors to generate the data. We test GOLF processes with {{a}} different number of factors, same or different range parameters. 
\label{eg:compare_diff_factors}
\end{example}

 In  Example \ref{eg:compare_diff_factors}, 
the factor processes are assumed to have the Mat{\'e}rn kernel in (\ref{equ:matern_5_2}) and unit variance. The kernel parameter is shared in the first scenario, where $\gamma_0=1/4$ and $\gamma_l=1/2$, and in the second scenario $\gamma_0=1/3$ and $\gamma_l=1/l$, for $l=1,...,d$. We estimate these parameters through the posterior samples from the MCMC algorithm. 

\begin{figure}[t]
\centering
  \begin{tabular}{ccc}
	\includegraphics[height=.3\textwidth,width=.33\textwidth]{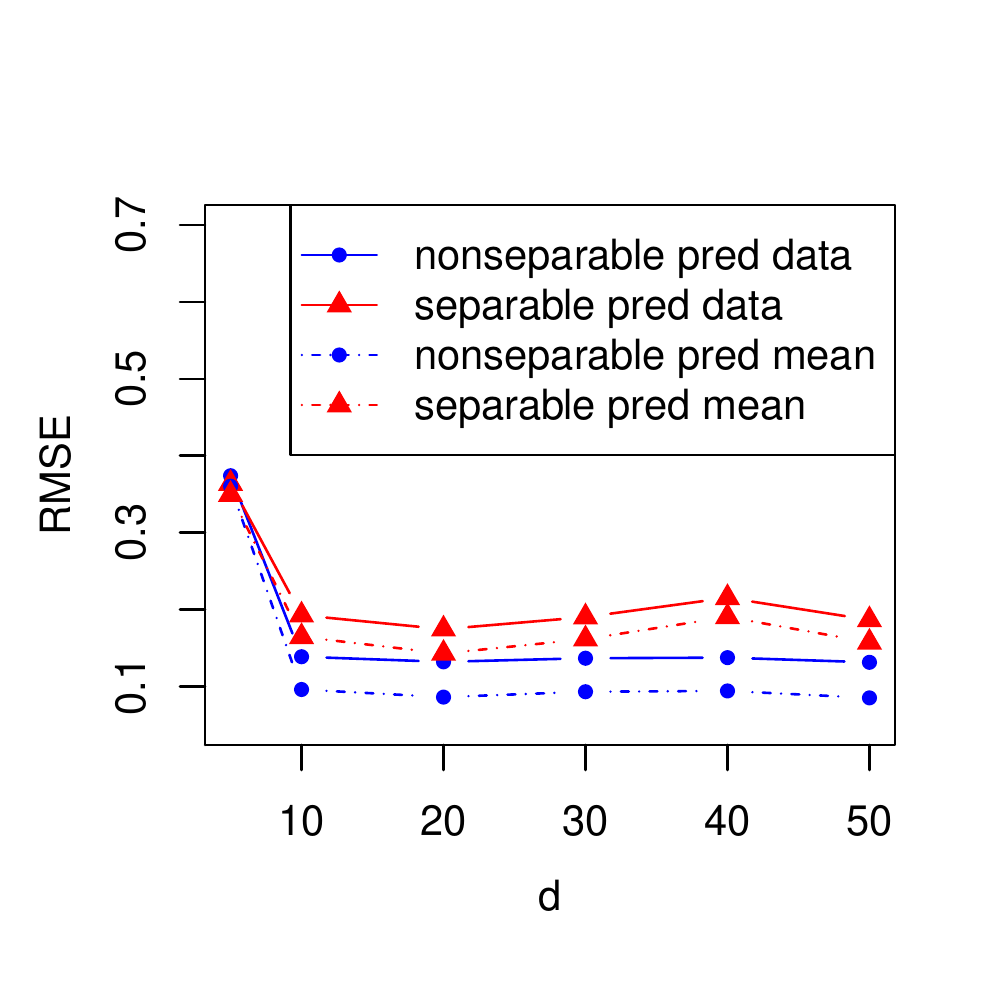}
	\includegraphics[height=.3\textwidth,width=.33\textwidth]{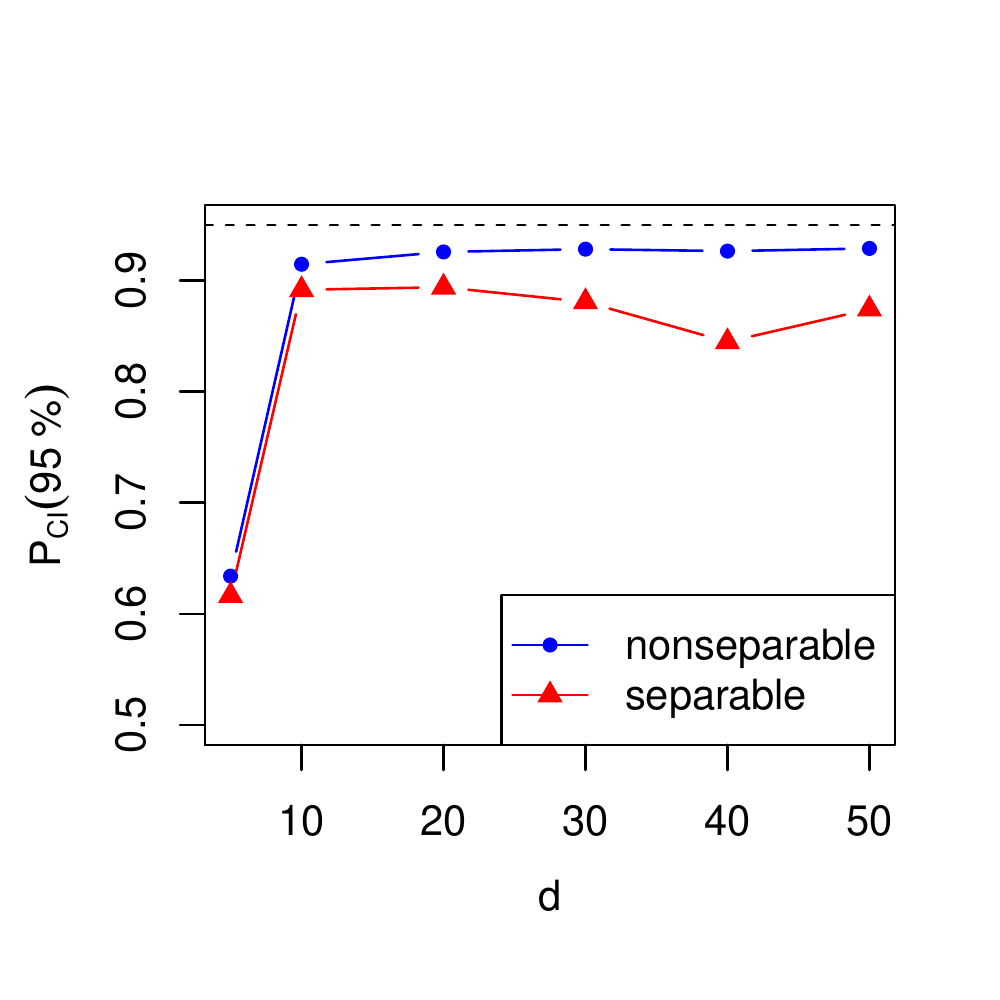} 
	\includegraphics[height=.3\textwidth,width=.33\textwidth]{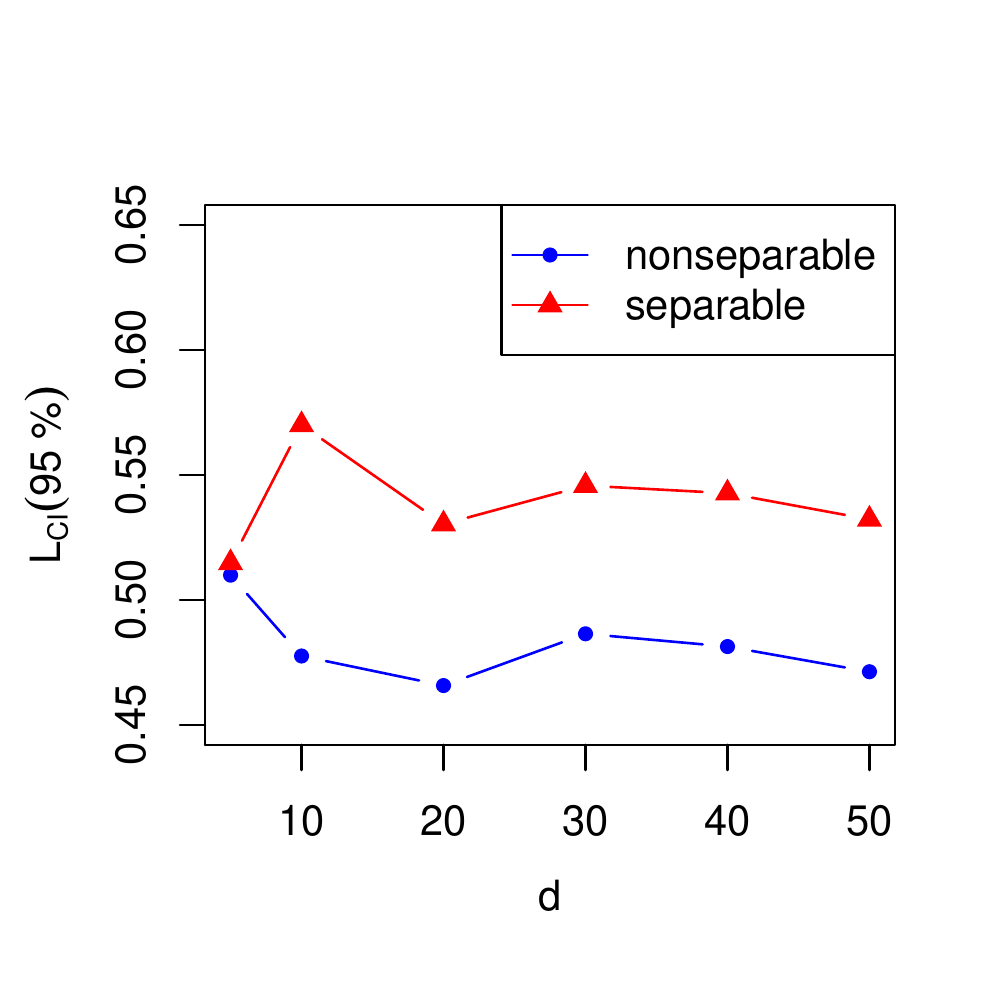}   \vspace{-.4in}
 \\
		\includegraphics[height=.3\textwidth,width=.33\textwidth]{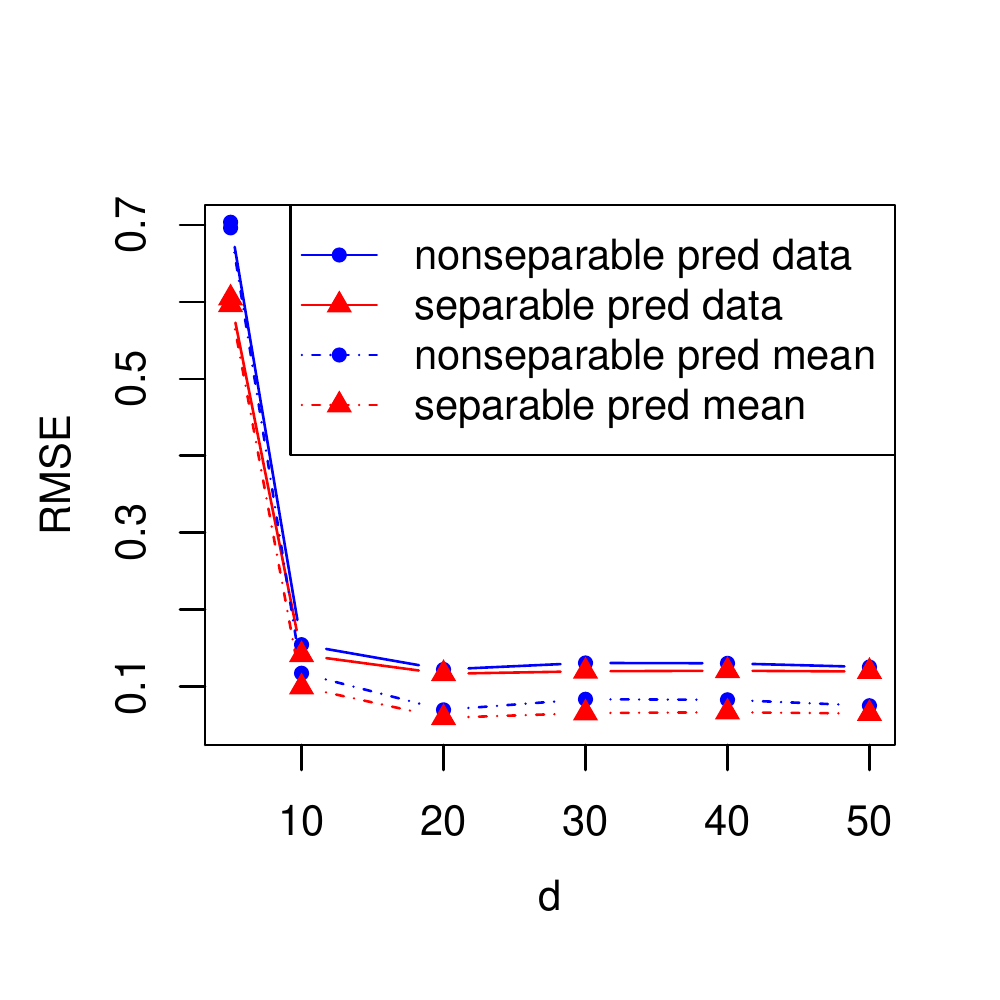}
	\includegraphics[height=.3\textwidth,width=.33\textwidth]{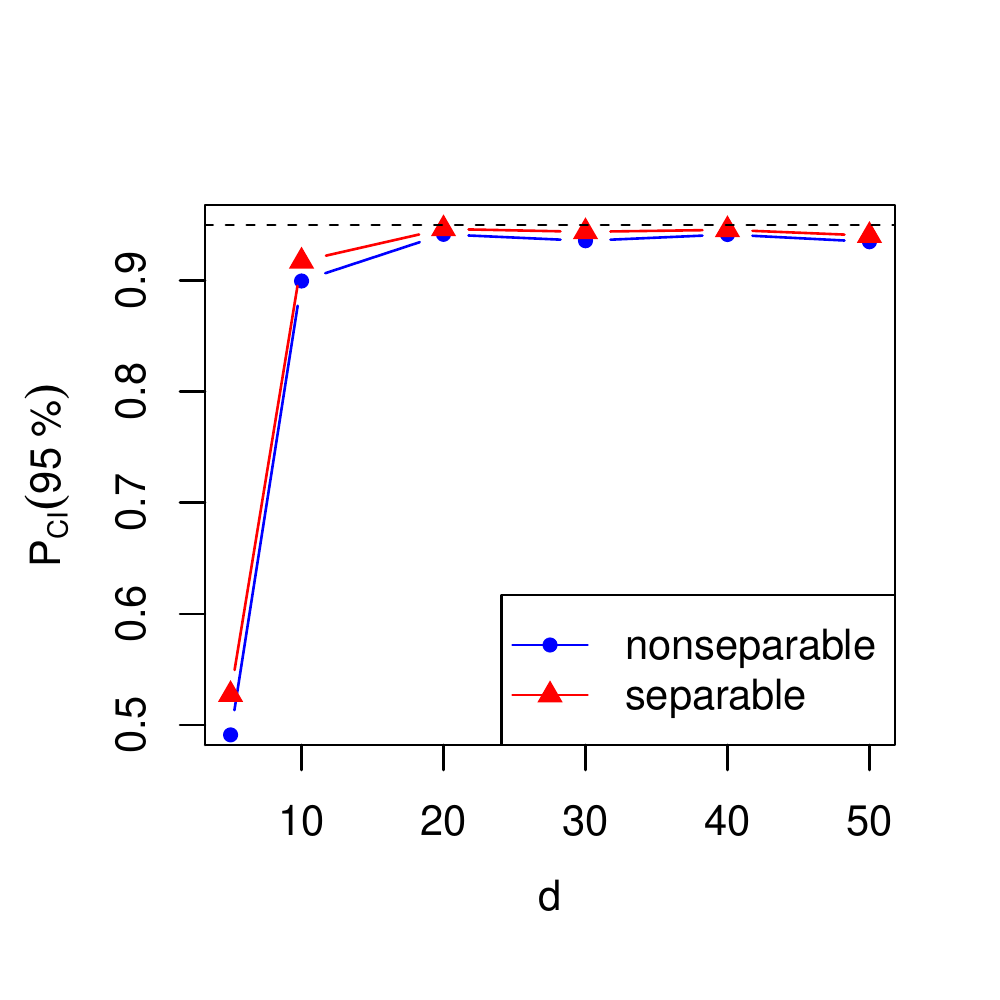} 
	\includegraphics[height=.3\textwidth,width=.33\textwidth]{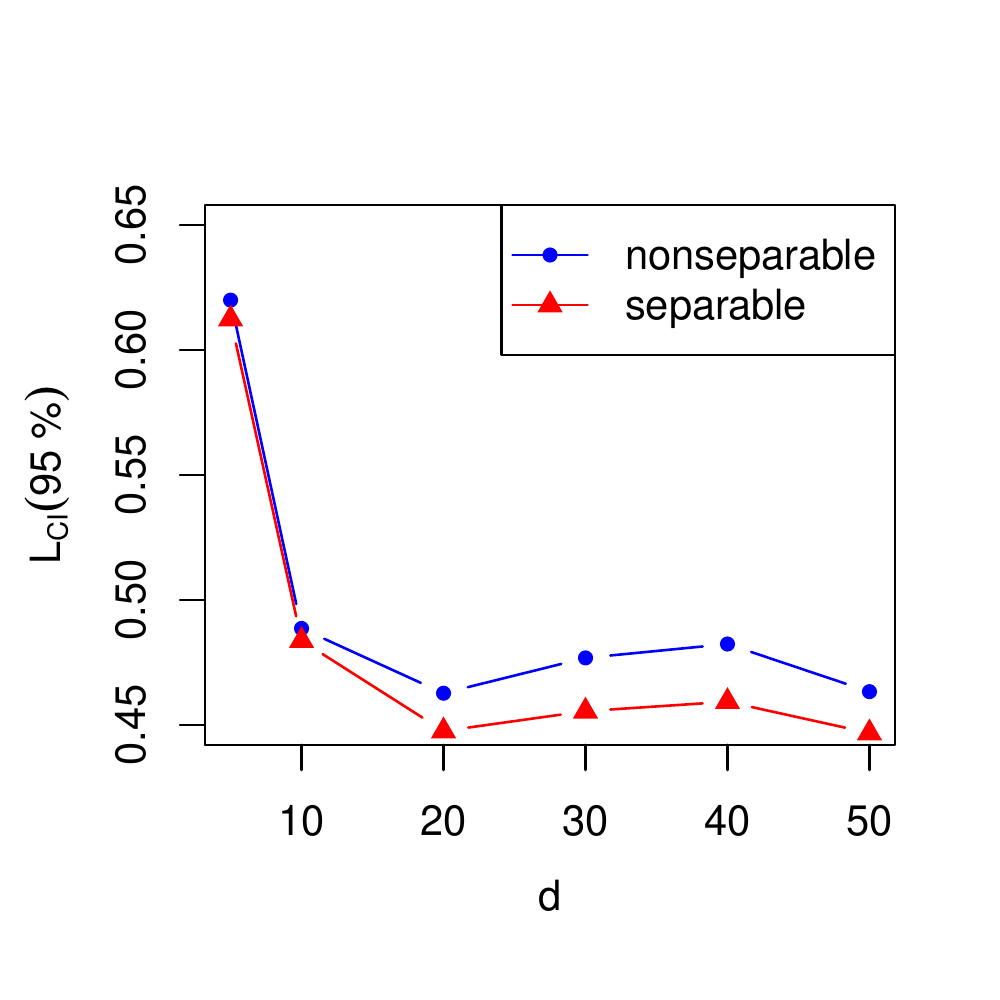}
  \end{tabular}
  \vspace{-.3in}
   \caption{The predictive performance of GOLF process with $d=5,10,20,30,40$ and $50$ factors for Example \ref{eg:compare_diff_factors}. 
   in the first row of panels, kernel parameters are different in simulating the data, whereas the parameters are the same in  for simulation in the second row of panels. 
   Blue curves and red curves denote the GOLF processes with  different kernel parameters and the same kernel parameter, respectively. In the left panels, the solid curves denote the RMSE for predicting the (noisy) observations, and the dashed curve denotes the RMSE for predicting the mean of the observations.  Proportions of observations covered in the $95\%$ predictive interval and the average length of the predictive interval are graphed in the middle and right panels, respectively.  }
   
   
\label{fig:pred_compare_diff_factors}

\end{figure}


Predictive performance of different approaches for data simulated by $d=30$ latent processes are graphed in Figure  \ref{fig:pred_compare_diff_factors}.
In the first row of the panels, since data are simulated by GOLF processes with different kernel parameters, nonseparable GOLF processes have smaller predictive RMSE and a shorter interval that covers almost $95\%$ of the data. In {{the}} second row of the panels,  GOLF processes with the same kernel parameter  seem to be slightly better, as the true factor process has the same kernel parameter. The difference between {{the}} two methods in the second row is smaller, as the GOLF model with a separable kernel {{is}} a special case of the one with different kernel parameters.


  \begin{figure}[t]
\centering
  \begin{tabular}{ccc}
  \hspace{-.15in}
	\includegraphics[height=.38\textwidth,width=.33\textwidth]{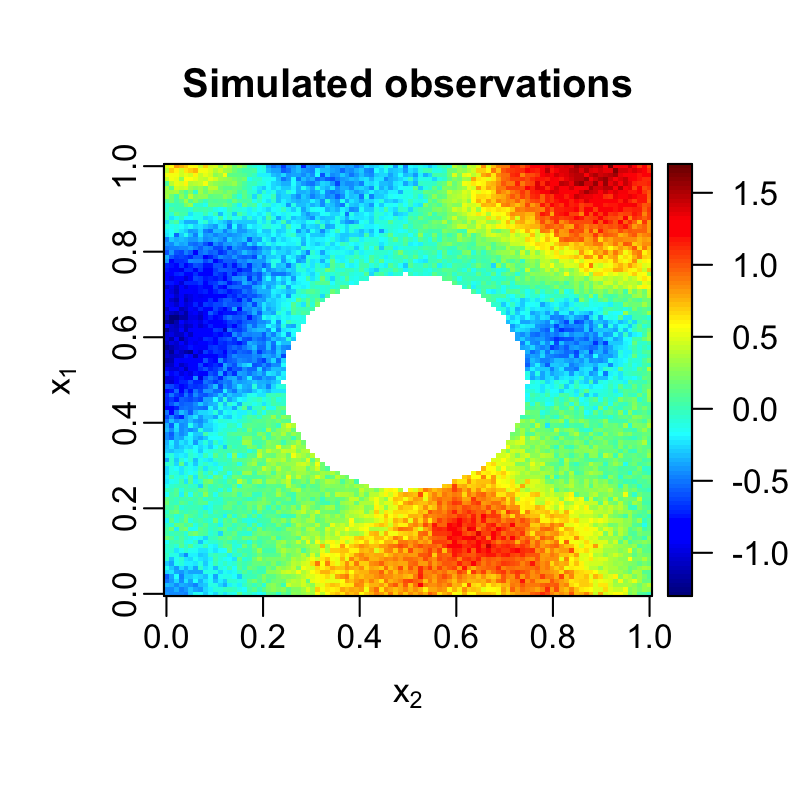} \hspace{-.12in}
	\includegraphics[height=.38\textwidth,width=.33\textwidth]{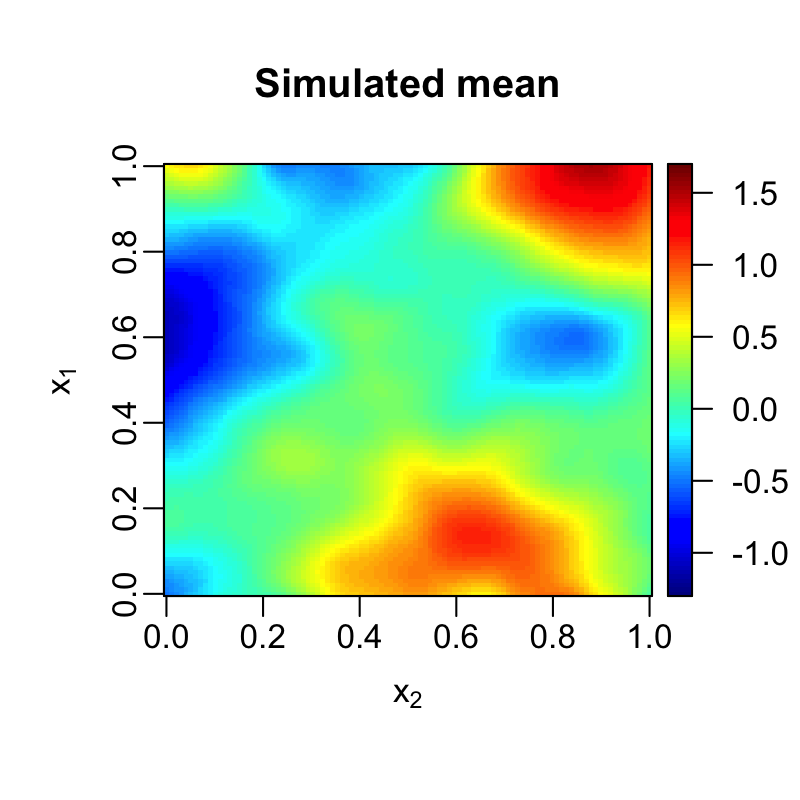} 
	\hspace{-.12in}
	\includegraphics[height=.38\textwidth,width=.33\textwidth]{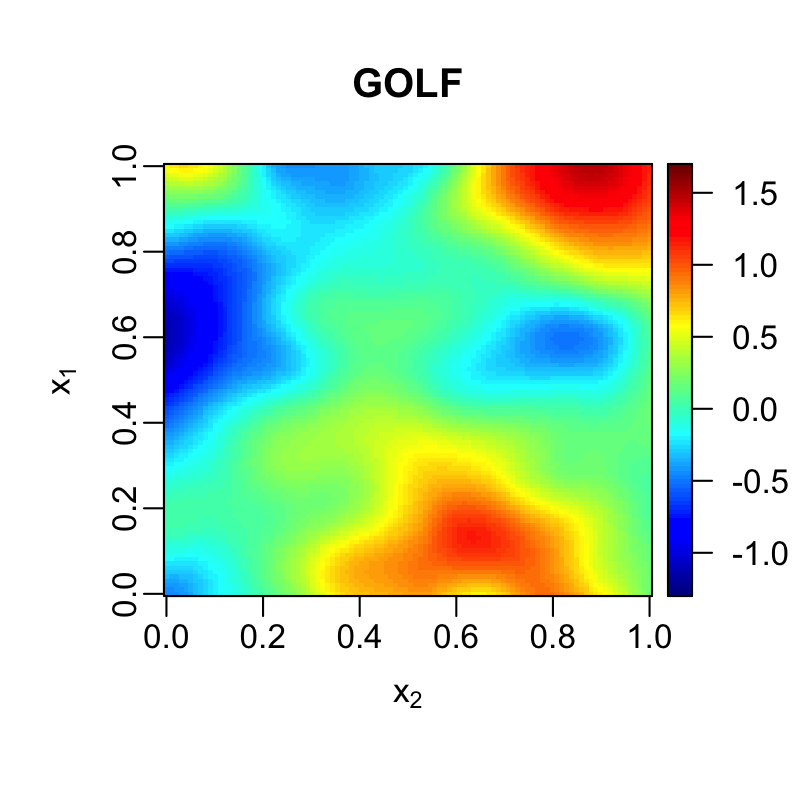} 
  \end{tabular}
  \vspace{-.25in}
   \caption{
   {The left figure shows the observed data in one simulation of Example \ref{eg:compare_diff_factors}, where a disk of observations is missing. The middle figure contains the mean of the data and the right figure is the prediction from GOLF process.}}
\label{fig:pred_eg2}
\end{figure}


 From Figure  \ref{fig:pred_compare_diff_factors}, we found that when we use $d=20$ factor processes or more, the predictive results seem to be similar, as the data are simulated using $d=30$ factor processes. The way of selecting the number of factors is currently ad-hoc. One may select the number of factors to ensure a large proportion  of the variance explained by the sum  of the eigenvalues of {{the}} correlation matrix $\mathbf R_s$. This simulation suggests that using more factors may be better in prediction than using very few factors.
 
 In Figure \ref{fig:pred_eg2}, we graph the simulated observations,  simulated mean,  and the prediction from the GOLF model with $d=30$ in one simulation.  Predictions look reasonably accurate.  Results  when the data are generated by a full rank kernel ($d=100$) are provided in Figure S3 in supplementary materials.  Results are very similar to Figure  \ref{fig:pred_compare_diff_factors}.









\section{Real applications}
\label{sec:real_data}
\subsection{Predicting large spatial data on an incomplete lattice}
\label{subsec:real_spatial_data}
We compare  GOLF processes with different approaches on predicting the missing temperature values in  \cite{heaton2019case}. This data set contains daytime land surface temperatures  on August 4, 2016, at $300\times 500$ spatial grids with the latitude and longitude  ranging from  34.30 to 37.07, and from -95.91 to -91.28, respectively.  The complete data set consists of  148,309 observations with $1,791$ missing values due to cloud cover. The  training data (plotted in the  left  panel in Figure \ref{fig:Temp_obs_EDA}) consists of  105,569 observations, whereas 42,740 observations were held out as the test data. Training observations and  full observations are graphed in  the upper panel in Figure \ref{fig:heaton_pred}.


	  \begin{table}[t]
	  \centering
	\begin{tabular}{lcccc} \hline
		Methods  & RMSE & $P_{CI}(95\%)$ & ${L_{CI}(95\%)}$ & Run time (mins) \\ \hline
		FRK & 3.16 & 0.77 &6.09 &3.53  \\
		Gapfill &1.86 &0.35  & 1.44 &6.98 \\
	{GOLF}   & {1.46} & {0.92} & {4.95} & {48.6} \\
		LAGP   & 2.07 &0.84& 5.70& 3.76\\
		LatticeKrig  & 1.68 &0.963  &6.58 &214.25 \\
		MRA    &1.85 &0.92 &5.54 &4.99 \\
		NNGP  & 1.64 & 0.95&5.84&1.14 \\
		Partition  &1.80 &0.82 &4.56 &827.37 \\
		SPDE    & 1.55  & 0.97  & 7.87 &34.8 \\
		\hline
	\end{tabular}
		 \caption{Comparison for the dataset in \cite{heaton2019case}. The standard deviation of observations  is $4.07$. 
		 For each method, we compute RMSE,  $P_{CI}(95\%)$ and $L_{CI}(95\%)$ defined in (\ref{equ:RMSE})-(\ref{equ:LCI}). A satisfying method should have small RMSE and small $L_{CI}(95\%)$  and $P_{CI}(95\%)$ closed to be $95\%$ nominal level. 
		 We compare the fixed rank kriging (FRK) (\cite{cressie2008fixed}), the Gapfill method (\cite{gerber1605predicting}), GOLF processes,  the local approximate Gaussian processes (LAGP) (\cite{gramacy2015local}), the lattice kriging (LatticeKrig) (\cite{Nychka2017LatticeKrig}), the multiresolution approximation (MRA) (\cite{katzfuss2017multi}),  the nearest neighbor Gaussian processes (NNGP) (\cite{datta2016hierarchical}), the spatial partitioning (Partition) (\cite{HeatonNonstationary2017}), and  stochastic partial differential equations (SPDE) (\cite{lindgren2011explicit}). } 
	\label{tab:comparison_pred_Heaton_data}

	\end{table}


We define GOLF processes on this dataset with $s$ being latitude and $x$ being longitude. Since areas with higher latitude typically have lower temperature on average, we assume a mean parameter for each latitude value, i.e. $\mathbf M= (\mathbf H_2 \mathbf B_2)^T$, where $\mathbf H_2=\mathbf 1_{n_2}$ and $\mathbf  B_2=(b_{2,1},...,b_{2,n_1})^T$. 
We let $d=n_1/2$ and use exponential kernels  with distinct variances and range parameters sampled from the marginal posterior distribution for GOLF processes. We compute $M=6000$ posterior samples where the first $20\%$ were used as the burn-in samples. Results of longer MCMC chains and  different initial values of the parameters are given in the supplementary materials.  

In \cite{heaton2019case}, 12 groups of researchers across the globe  implemented  their  methods to predict missing temperature values for competition.  Among this cohort of researchers are authors that conjured up some of the most popular methods for large spatially correlated data. 
 Other than GOLF processes, we implement 8 of 12 approaches based on the code provided in \citep{heaton2019case}. We could not implement the other 4 approaches due to memory limitation of {{the}} computing facility or unavailability of the code. All computations are operated on a 3.60GHz 8 cores Intel i9 processor with 32 GB of RAM on a macOS Mojave operating system. 






The predictive performance of different approaches {{is}} recorded in  Table \ref{tab:comparison_pred_Heaton_data}. 
Most of {{the}} results are  consistent with what is shown in \cite{heaton2019case}, whereas small differences remain for those requiring random starts or stochastic algorithms. E.g., 5 implementation of the SPDE method 
gives different RMSE  ranging from 1.55 to 1.88. 
Besides,  running time of some methods are  slightly different. 
For SPDE and LatticeKrig, for instance, it  takes 35 mins and 214 mins to run in our system, respectively, whereas it takes 138 mins and 78 mins to run in \cite{heaton2019case}, respectively.


We acknowledge that held-out observations were not released in \cite{heaton2019case}, adding difficulty for model specification.  The good performance {{of}} the GOLF model may be explained by two reasons. First, different mean parameters are assumed at each latitude, which is more flexible to capture information from a large number of  observations. 
Second, we assume different  range  and variance parameters of the factor processes, which {{are}} more flexible than  the separable  or isotropic kernel functions.  

The $95\%$ predictive interval of the GOLF model is the shortest, and it covers around $92\%$ of the held out test data, as shown in Table \ref{tab:comparison_pred_Heaton_data}. 
In  supplementary materials,  we provide diagnostic plots of the fitted values from the GOLF model and predictive performance based  on several configurations, including $40,000$ MCMC samples and different initial parameters. The predictive performance of the GOLF model at different configurations is similar. Besides, the computational time of GOLF per one MCMC iteration is around $0.49s$ for this example, which is comparable to NNGP (0.53s) and faster than MRA (3.29s) for  one iteration. The posterior sampling obtained here provided uncertainty quantification of model parameters, whereas most of the methods provided in Table \ref{tab:comparison_pred_Heaton_data} only provide a point estimator of the  parameters. Future works are needed to reduce the number of iterations in GOLF to achieve {{a}} similar level of predictive accuracy. 

    \begin{figure}[t]
\centering
  \begin{tabular}{ccc}
  \includegraphics[height=.25\textwidth,width=.45\textwidth ]{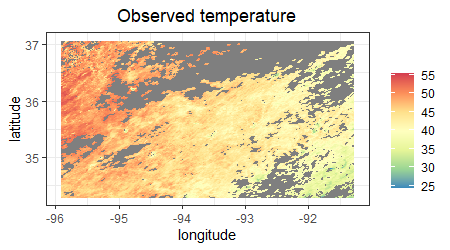}
	\includegraphics[height=.25\textwidth,width=.45\textwidth]{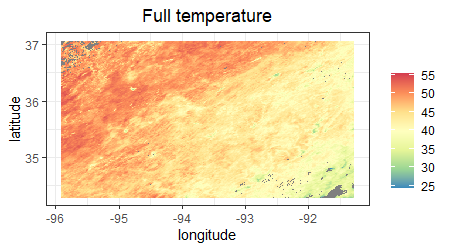} \hspace{-.1in}\\  \vspace{-.4in}
	\includegraphics[height=.25\textwidth,width=.45\textwidth]{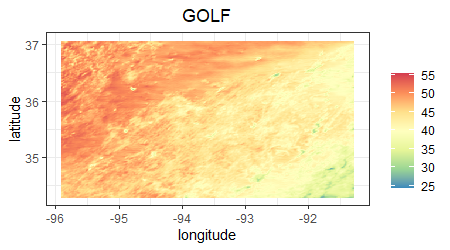} 
	\includegraphics[height=.25\textwidth,width=.45\textwidth]{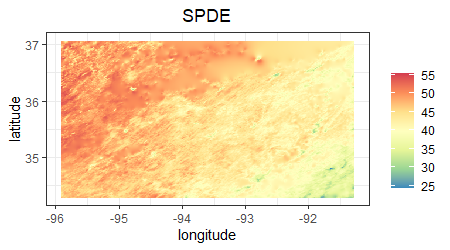} 
  \end{tabular}
  \vspace{.25in}
   \caption{
   {The top panels show the observed temperature and full temperature, respectively, where the gray area contains unobservable points. The bottom panel are the predictions from two methods, GOLF and SPDE, respectively.}
   }
   
\label{fig:heaton_pred}
\end{figure}

The predictive mean of the GOLF processes and SPDE  are graphed in {{the}} middle panel and right panel in Figure \ref{fig:heaton_pred}, respectively. Predictions from the GOLF processes are more accurate for predicting temperatures in areas with high latitude, possibly due to flexible mean parameters estimated from data.  Both methods seem to be slightly oversmoothing. Yet predicting the missing values of this data set is challenging, as the observations are missing in  spatial blocks. Both methods seem precise in prediction. 

\subsection{Analysis of large spatio-temporal data set}
We consider the monthly gridded temperature anomalies from U.S. National Oceanic and Atmospheric Administration (NOAA) \footnote{\url{ftp://ftp.ncdc.noaa.gov/pub/data/noaaglobaltemp/operational}}. The data set contains the average air and marine temperate anomalies at 5 degrees longitude-latitude grids with respect to 1981-2010 base period.  R code and examples to load NOAA gridded data can be found in \cite{SShen2017}. We compare the predictive performance using the data from Jan 1999 to Dec 2018. For each month, we observe the temperature anomalies at $n_1=36\times 28$ spatial grids with longitude ranging from 182.5 to 357.5 and with latitude ranging from -62.5 to 72.5, respectively.  There are $11,122$ missing data, leaving the total number of observations to be $230,798$. We held out $50\%$ randomly sampled temperature anomalies  as the missing data, and the rest $50\%$ is used as training data (i.e., $n=n^*=115,399$). Predicting the missing values in this scenario is more difficult  than the example in \citep{gu2018generalized}, where the data are missing in a set of locations over the same  months. 


We fit the GOLF processes with the covariance of each spatial coordinate modeled by the Mat{\'e}rn covariance, and the factors processes are defined on the temporal input with different kernel parameters. 
Due to computational limitation, we let the number of factors  be $d=0.75^2 n_1=567$ and assume the factor loadings to be a Kronecker product of the first {{three-quarters}} of the eigenvectors of the sub-covariance matrices for  longitude and latitude. Although we have a large number of factors, the computational complexity is $O(Nd_{max})$ with $d_{max}=0.75 \times 36=48$ rather than $O(Nd_1d_2)$ by the mode multiplication of tensor (see Section \ref{subsec:complexity} for the discussion).  We assume  the coefficients of the intercept and linear coefficients are different at each location, i.e. $\mathbf M=(\mathbf H_2 \mathbf B_2)^T$ where $\mathbf H_2=[\mathbf 1_{n_2}, \mathbf x]$, with $ \mathbf x$ being $240$ months  and $\mathbf B_2$ being a matrix of $ 2 \times n_1$ coefficients. We use $M=3000$ MCMC samples with the first $20\%$ as the burn-in samples, as posterior samples  converge at a small number of iterations in this example. 

	  \begin{table}[t]
	  \centering
	{\begin{tabular}{lcccc} \hline
		Methods  & RMSE & $P_{CI}(95\%)$ & ${L_{CI}(95\%)}$ & Run time (mins) \\ \hline
		FRK & 0.846  & 0.967 &3.92 &29.4  \\
	   {GOLF}   & {0.325} & {0.942} & {1.08} & {43.9} \\
	    LAGP &0.695 & 0.951&1.80 &6.18 \\
		Spatial model 1   & 0.365 &0.928 &2.09 &26.5 \\
		Spatial model 2   & 0.348 &0.928 &2.02 &42.7\\

		\hline
	\end{tabular}
	}
		 \caption{Predictive performance of different approaches for the NOAA monthly gridded temperature dataset. The standard deviation of the outcomes in this dataset is  $0.940$. Results of the FRK, GOLF,  and LAGP are given in the first to the third rows. For the results in the fourth and fifth rows,  spatial models were fitted using the $\sf RobustGaSP$ package  with one initial value and two initial values of the range and nugget parameters for finding their marginal posterior mode, respectively.
	\label{tab:comparison_NOAA_data}}

	\end{table}
	
	In Table \ref{tab:comparison_NOAA_data}, we compare the GOLF processes with a few other spatial and spatio-temporal methods for the NOAA dataset. We fit two spatial models separately for each month using the {$\sf RobustGaSP$} package available on CRAN.  
	Also implemented are FRK and  LAGP based on  their packages \citep{zammit2017fixedrank,gramacy2016lagp}. 
	
	As shown in Table \ref{tab:comparison_NOAA_data},  GOLF processes have the smallest predictive RMSE and {{the}} shortest predictive interval that covers around $94\%$ of the held-out output. Since the temporal input is not used, it is not surprising that the RMSE and the  length of {{the}} predictive interval of the two spatial models are larger than the ones by  GOLF processes. If we include the temporal inputs, the computation cost is too {{large}} for inverting the covariance matrix directly. FRK and LAGP also seem to have a larger predictive error, though both the spatial and temporal inputs are used in these methods. 
	
	Predictions from GOLF processes are more accurate due to three  reasons. 
	First,  
	we can compute the model with a large number of factors efficiently, and no further approximation of the likelihood function is required. Second, mean and trend parameters at each location are  different, making the model flexible to capture the dynamic trend of temperature values at different locations.  Finally,  Latent factor processes have different kernel parameters that fit diverse  smoothness levels of projected observations. 

 In Figure \ref{fig:Temp_NOAA}, we graph the full temperature anomalies in Jan 2018, predictions from the GOLF and spatial GP model by {$\sf RobustGaSP$} package. $50\%$ of the observation in the left panel are held out for testing. Both models seem to be accurate. Since the temporal coordinate is used in  prediction, the predictive error by  GOLF processes is  smaller.

	\begin{figure}[t]
		\begin{tabular}{lll}
		\includegraphics[height=.25\textwidth,width=.33\textwidth ]{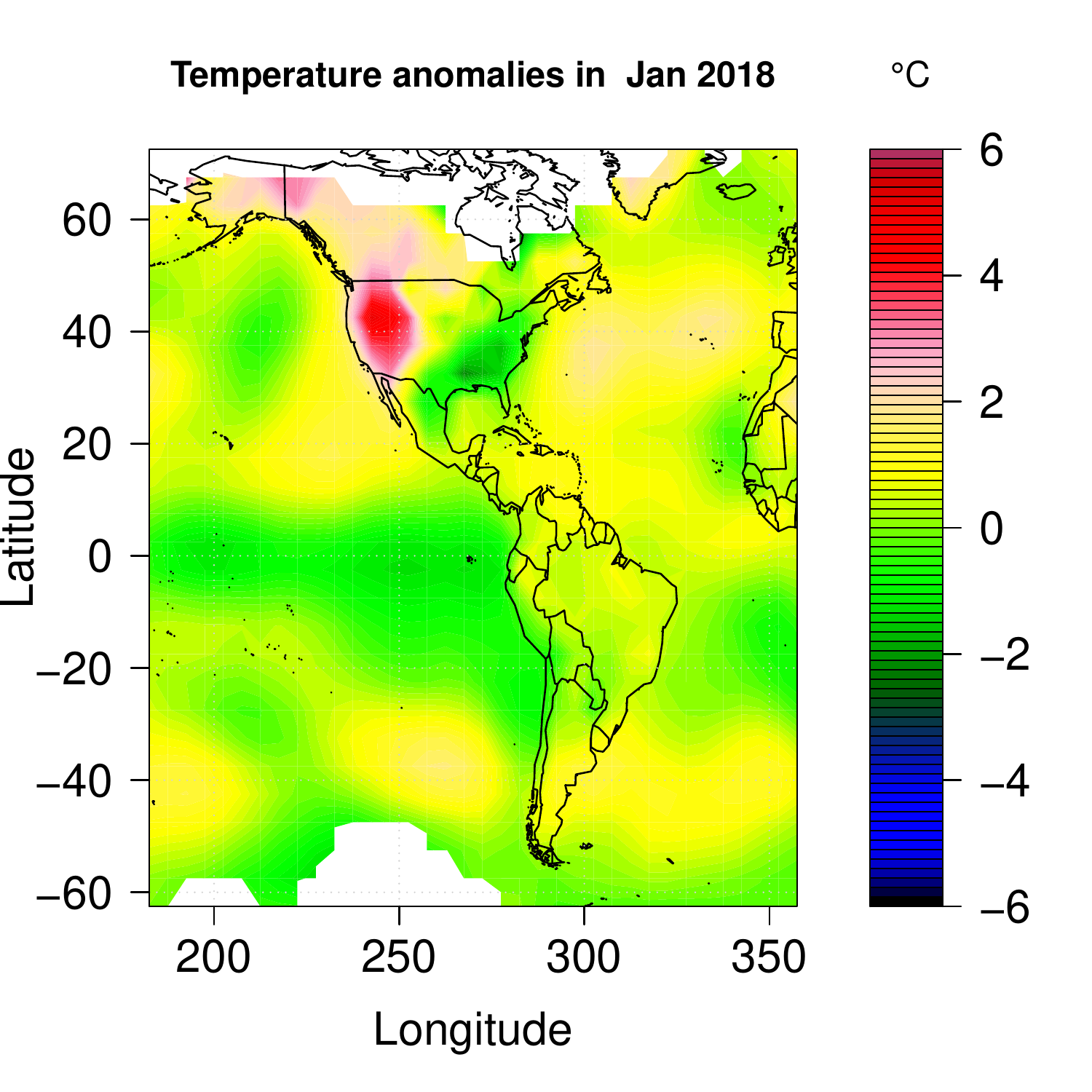}
		\includegraphics[height=.25\textwidth,width=.33\textwidth ]{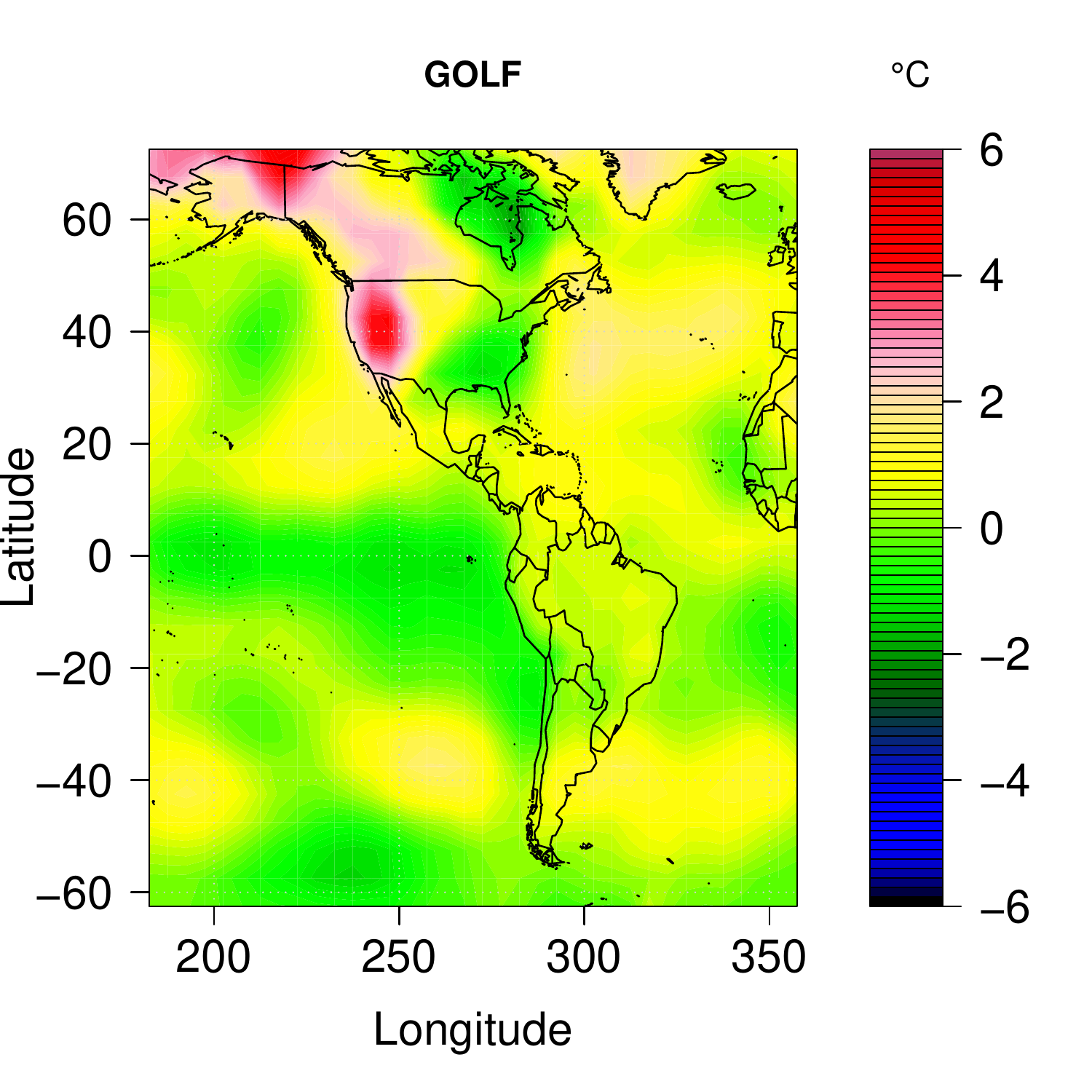}
		\includegraphics[height=.25\textwidth,width=.33\textwidth ]{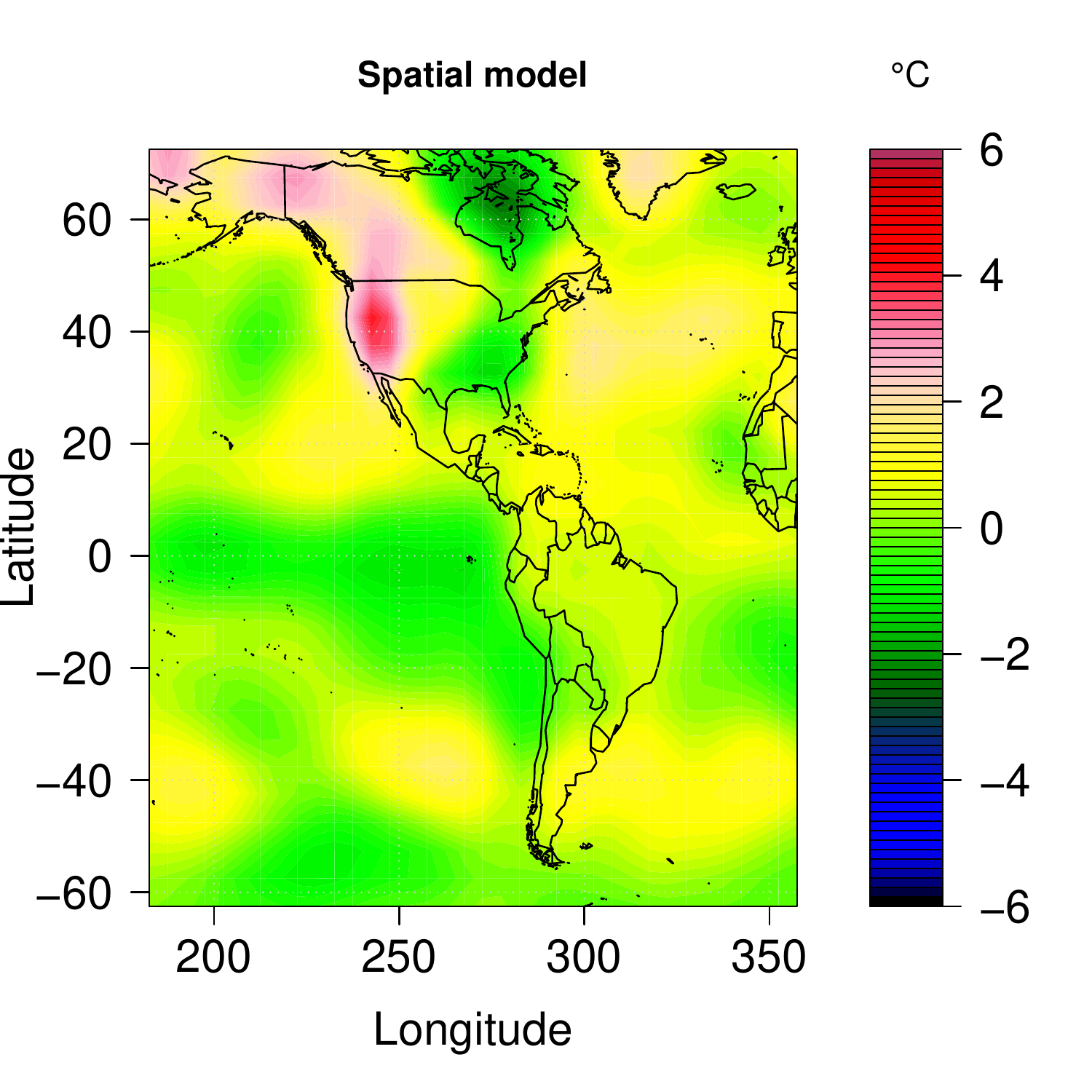}
		\vspace{-.12in}

		\end{tabular}
	\caption{Full temperature {anomalies in Jan 2018},  predictions  by the GOLF model and  the spatial model by {$\sf RobustGaSP$} package  are shown in  left, middle and right panels, respectively.  }
	\label{fig:Temp_NOAA}
\end{figure}

	
	





\section{Concluding remarks}
\label{sec:conclusion}
We have introduced  GOLF processes as a  computationally feasible approach to model  large incomplete lattice observations. For GPs with a product covariance function or LMC with orthogonal latent factor loadings,  the likelihood can be decomposed into a product of multivariate normal densities, and  prior independence of factor processes leads to posterior independence of factor processes. These two properties allow one to reduce the computational burden of GPs on incomplete lattice observations without approximating the likelihood function. Further computational reduction can be made by reducing the number of factors as well. Besides, we have introduced a flexible way to model the mean function and  the closed-form marginal likelihood is derived to alleviate the identifiability issue. Finally, we have developed an MCMC algorithm for Bayesian inference for 
large incomplete matrices of spatial and spatio-temporal data. 

The computational tools developed in this work require observations from a lattice with potential missing values. Approximation methods such as the NNGP approach may be integrated to model correlated data with a more general design. 
Besides, further computational reduction can be made by reducing the number of factors, and  a principle way to select the number of factors will be useful. 
Finally,  direct marginalization of factor processes based on an elementwise representation of GPs 
may be feasible to reduce the computation time from drawing a large number of posterior samples.





~

{\bf Acknowledgements}. We thank the editor, associate editor and referee for their comments that substantially improved the article. This research was  supported by National Science Foundation  under Award Number DMS-2053423 and National Institutes of Health under Award Number R01DK130067. We thank the editor, associate editor and referee for their great suggestions that substantially improve  this manuscript. 


\bibliographystyle{apalike}
\bibliography{References_2019}
%

\newpage

\begin{center}
    {\bf \LARGE Supplementary materials}
\end{center}

\beginsupplement

\quad 

This supplementary materials contain three parts. The proof of Section \ref{sec:GOLF} is given in Section \ref{sec:proof}. The additional numerical results for the simulated studies and real applications are given in Section \ref{sec:supp_simulation}
 and  Section \ref{sec:supp_real_data}, respectively.

\section{Proofs for Section \ref{sec:GOLF}}
\label{sec:proof}
\subsection{Auxiliary facts}
\begin{enumerate}[1.]
\item Let $\mathbf A$ and $\mathbf B$ be matrices, 
\[ (\mathbf A \otimes \mathbf B)^T=  (\mathbf A^T  \otimes \mathbf B^T); \]
further assuming $\mathbf A$ and $\mathbf B$ are invertible, 
\[ (\mathbf A \otimes \mathbf B)^{-1}=  \mathbf A^{-1} \otimes  \mathbf B^{-1}. \]
\label{item:matrix_kronecker_trans_inv}
\item Let $\mathbf A$, $\mathbf B$, $\mathbf C$ and $\mathbf D$ be the matrices such that the products $\mathbf A\mathbf C$ and $\mathbf B\mathbf D$ are matrices,  
\[(\mathbf A \otimes \mathbf B)(\mathbf C\otimes \mathbf D)= (\mathbf A \mathbf C) \otimes  (\mathbf B \mathbf D).\]
\label{item:matrix_kronecker_otimes}
\item For matrices $\mathbf A$,  $\mathbf B$ and  $\mathbf C$, 
\begin{align*}
(\mathbf C^T \otimes \mathbf A) \mbox{vec}(\mathbf B)= \mbox{vec}(\mathbf A \mathbf B \mathbf C);
\end{align*}
further assuming $\mathbf A^T \mathbf B$ is a matrix, 
\begin{align*}
\tr(\mathbf A^T \mathbf B )= \mbox{vec}(\mathbf A)^T \mbox{vec}(\mathbf B).
\end{align*}
\label{item:vectorization}
\item For any invertible $n \times n$ matrix $\mathbf C$,
\[|\mathbf C+\mathbf A \mathbf B|=|\mathbf C||\mathbf I_n + \mathbf B \mathbf C^{-1} \mathbf A |. \]
\label{item:determinant}
\end{enumerate}

\subsection{Proofs for Section \ref{subsec:orthogonal_independence}}


The following denotation are used in the proof: $\mathbf {Y}_{-M}=\mathbf Y-\mathbf M$, $\mathbf { Y}_{v,-M}=vec(\mathbf Y-\mathbf M)$, $\mathbf Z_{vt} =vec(\mathbf Z^T)$ and $\mathbf A_v=[\mathbf I_{n_2} \otimes \mathbf a_1,..., \mathbf I_{n_2} \otimes \mathbf a_d]$. Let $\bm \Sigma_v$ be an $n_2 d\times n_2 d$ matrix where the $l$th diagonal block is $\bm \Sigma_l$. Denote $\etr(.)=\exp(\tr(.))$.  

\begin{proof}[Proof of Equation \ref{equ:marginal_lik}]
Denote $C_Y=(2\pi \sigma^2_0)^{-\frac{n_1 n_2}{2}  }   \prod^d_{l=1} \left|\bm \Sigma_l/\sigma^2_0+ \mathbf I_{n_2} \right|^{-1/2}$. 
Directly marginalizing out $\mathbf Z$, one has
\begin{align*}
&p(\mathbf Y \mid  \bm{\Theta}  )\\
=& C_Y \exp\left(-\frac{  \mathbf {Y}^T_{v ,-M}\left(\mathbf{I}_{n_1 n_2} - \sum_{l=1}^{d} (\sigma_0^2 \bm{\Sigma}_l^{-1} + \mathbf{I}_{n_2})^{-1} \otimes (\mathbf{a}_l \mathbf{a}_l^T) \right)  \mathbf {Y}_{v ,-M}}{2\sigma^2_0} \right) \\
=& C_Y\exp\left(- \frac{ \mathbf {  Y}^T_{v,-M} \mathbf { Y}_{v,-M} -  \mathbf {Y}^T_{v ,-M}\sum^d_{l=1}\mbox{vec}( \mathbf a_l \mathbf a^T_l \mathbf {Y}_{v ,-M}\bm (\sigma^2_0 \bm \Sigma^{-1}_l+ \mathbf I_{n_2})^{-1} ) }{2\sigma^{2}_0}\right) \\
=&C_Y\etr\left(- \frac{ \mathbf { Y}^T_{-M} \mathbf { Y}_{-M} -  \sum^d_{l=1}\mathbf { \tilde y}_{l} \mathbf { \tilde y}^T_{l} \bm (\sigma^2_0 \bm \Sigma^{-1}_l+ \mathbf I_{n_2})^{-1}   }{2\sigma^{2}_0 } \right)\\
=& C_Y \exp\left( -\frac{\sum^d_{l=1} \mathbf { \tilde y}^T_{l} (\bm \Sigma_l/\sigma^2_0 + \mathbf I_{n_2})^{-1}\mathbf { \tilde y}_l +  \sum^{n_1}_{l=d+1}\mathbf { \tilde y}^T_{l}\mathbf { \tilde y}_{l}
}{2\sigma^{2}_0} \right),
\label{equ:last_line_pf}
\nonumber 
\end{align*}
where the first equation is based on  Lemma 1 and the Woodbury matrix identity (to compute the normalizing constant $C_Y$); the second and third equations are from fact \ref{item:vectorization}; the fourth equation is from Woodbury matrix identity. The Equation (\ref{equ:marginal_lik}) follows immediately.

\end{proof}

\begin{proof}[Proof of Corollary \ref{cor:ind_post}]

The proof is implied by the proof of Theorem 4 in  \citep{gu2018generalized}. For completeness of this article, we include the proof  below.

From Equation (\ref{equ:GOLF_model}) and Equation (\ref{equ:Z}), we have 
\begin{align*}
p(\mathbf Z_{vt} \mid \mathbf Y, \bm { \Theta} )&\propto \exp\left( \frac{ (\mathbf {Y}_{v,-M} -\mathbf { A}_v \mathbf Z_{vt} )^T (\mathbf {Y}_{v,-M}  -\mathbf { A}_v \mathbf Z_{vt})}{2 \sigma^2_0}\right) \exp\left( -\frac{1}{2} \mathbf Z^T_{vt} \bm { \Sigma}^{-1}_v \mathbf Z_{vt} \right) \\
&\propto \exp\left\{ -\frac{1}{2}(\mathbf Z_{vt} -\bm \mu_{Z_{vt}} )^T \left( \frac{\mathbf { A}^T_v \mathbf { A}_v}{ \sigma^2_0 }+\bm { \Sigma}_v^{-1} \right) (\mathbf Z_{vt} -\bm \mu_{Z_{vt}})\right\},
\end{align*}
where $\bm \mu_{Z_{vt}} =  ({\mathbf { A}^T_v \mathbf { A}_v} +  \sigma^2_0 \bm { \Sigma}_v^{-1})^{-1}  \mathbf { A}^T_v \mathbf {Y}_{v,-M} $. Note  $\mathbf { A}^T_v \mathbf { A}_v=\mathbf I_{n_2 d} $, from which we have
\begin{equation}
\mathbf Z_{vt} \mid \mathbf Y, \bm { \Theta} \sim \mbox{MN}\left(\bm \mu_{Z_{vt}} , \,  \left(\frac{1}{ {\sigma}^2_0 } \mathbf I_{n_2 d}+\bm { \Sigma}_v^{-1}\right)^{-1} \right).
\label{equ:Z_posterior} 
\end{equation}

 
 Based on vectorization, one has 
\begin{align}
\bm \mu_{Z_{vt}}&= \left( {\begin{array}{*{20}{c}}
\left( \sigma^2_0 \bm {  \Sigma}^{-1}_1 +{\mathbf I_{n_2}}\right)^{-1}\otimes  \mathbf { a}^T_1  \\
\vdots \\
\left( \sigma^2_0 \bm {  \Sigma}^{-1}_d +{\mathbf I_{n_2}}\right)^{-1} \otimes  \mathbf { a}^T_d \\
 \end{array} } \right) \mbox{vec}({\mathbf Y})
 = \left( {\begin{array}{*{20}{c}}
\mbox{vec}\left(\mathbf { a}^T_1 \mathbf Y_{-M} \left( \sigma^2_0 \bm {  \Sigma}^{-1}_1 +{\mathbf I_{n_2}} \right)^{-1}\right)  \\
\vdots \\
\mbox{vec}\left(\mathbf { a}^T_d \mathbf Y_{-M} \left( \sigma^2_0 \bm {  \Sigma}^{-1}_d +{\mathbf I_{n_2}}\right)^{-1}\right)   \\
 \end{array} } \right) \nonumber\\
 &=\mbox{vec} \left( {\begin{array}{*{20}{c}}
\mathbf { a}^T_1 \mathbf Y_{-M} \left( \sigma^2_0 \bm {  \Sigma}^{-1}_1 +{\mathbf I_{n_2}} \right)^{-1}  \\
\vdots \\
\mathbf { a}^T_d \mathbf Y_{-M} \left( \sigma^2_0 \bm {  \Sigma}^{-1}_d +{\mathbf I_{n_2}}\right)^{-1}   \\
 \end{array} } \right)^T. 
 \label{equ:Z_hat_posterior}
 \end{align}
Note that the covariance matrix of $\bm \mu_{Z_{vt}}$ is a block diagonal matrix. The results follow by Equation (\ref{equ:Z_hat_posterior}) and the Woodbury matrix identity.

\end{proof}

\subsection{Proofs for Section \ref{subsec:mean_function}}
Note $\mathbf A_F=[\mathbf A_s,\mathbf A_c]=[\mathbf a_1,\mathbf a_2,...,\mathbf a_{n_1}]$, where $\mathbf A_c$ is an $n_1\times (n_1-d)$ matrix of the orthogonal complement of $\mathbf A_s$. We need the following lemma to prove Theorem \ref{thm:marginal_post_B_1_B_2}.

\begin{lemma}
After marginalizing out the factors $\mathbf Z$, we have the marginal posterior  distribution of  the transformed regression coefficients, 
\begin{enumerate}
\item (Marginal distribution of transformed row regression coefficients). Assume  $\mathbf M=\mathbf H_1 \mathbf B_1$ and the objective prior $\pi(\mathbf B_1)\propto 1$ for $\mathbf B_1$. Let $  \mathbf { \tilde B}_1=[\bm {\tilde b}_{1,1},...,\bm {\tilde b}_{1,n_1}]=\mathbf B_1^T \mathbf H_1^T \mathbf A_{F}$ be an $n_2\times n_1$ matrix of transformed coefficients. Assume the marginal posterior distribution of $\mathbf { \tilde B}_1$ follows
	\begin{equation}
	p(\mathbf { \tilde B}_1 \mid \mathbf Y, \bm \Theta_{-B_1})= \prod^{d}_{l=1} \mathcal{PN}(\bm {\tilde b}_{1,l}; \mathbf {\tilde y}_l, \bm {\tilde \Sigma}_l ) \prod^{n_1}_{l=d+1} \mathcal{PN}(\bm {\tilde b}_{1,l}; \mathbf {\tilde y}_l, \sigma^2_0 \mathbf I_{n_2}),  
	\label{equ:tilde_B1}
	\end{equation}
	where $\mathbf{\tilde y}_l$ is defined in equation (\ref{equ:marginal_lik}) and $\bm{\tilde \Sigma}_l$ is defined in corollary \ref{cor:ind_post}. Then we can sample $(\mathbf { B}_1\mid \mathbf Y, \bm \Theta_{-B_1})$ by $(\mathbf H^T_1 \mathbf H_1)^{-1}\mathbf H^T_1  \mathbf A_F \mathbf {\tilde B_1}^T$, where  $\mathbf {\tilde B_1}^T$ are sampled from the $p(\mathbf { \tilde B}_1 \mid \mathbf Y, \bm \Theta_{-B_1})$ in equation (\ref{equ:tilde_B1}). 
	
	

\item (Marginal distribution of transformed column regression coefficients). 
Assume $\mathbf M=(\mathbf H_2 \mathbf B_2)^T$ and the objective prior $\pi(\mathbf B_2)\propto 1$ for the regression parameters $\mathbf B_2$. Let $  \mathbf { \tilde B}_2=[\bm {\tilde b}_{2,1},...,\bm {\tilde b}_{2,n_1}]=\mathbf B_2 \mathbf A_{F}$ be a $q_2\times n_1$ matrix. The marginal posterior distribution of  $\mathbf {\tilde B}_2$ follows

	\begin{equation}
	p(\mathbf { \tilde B}_2 \mid \mathbf Y, \bm \Theta_{-B_2})= \prod^{n_1}_{l=1} \mathcal{PN}(\mathbf {\tilde b}_{2,l}; \bm \mu_{\tilde b_{2,l}}, \bm \Sigma_{\tilde b_{2,l}} ),  
	\label{equ:tilde_B2}
	\end{equation}
	where  $\bm \mu_{\tilde b_{2,l}}= (\mathbf H^T_2 \bm {\tilde \Sigma}^{-1}_l  \mathbf H_2)^{-1}\mathbf H^T_2\bm {\tilde \Sigma}^{-1}_l \mathbf {\tilde y}_l$ and $\bm \Sigma_{\tilde b_{2,l}}= (\mathbf H^T_2 \bm {\tilde \Sigma}^{-1}_l  \mathbf H_2)^{-1}$ for $l=1,...,d$; $\bm \mu_{\tilde b_{2,l}}=(\mathbf H^T_2\mathbf H_2 )^{-1} \mathbf H^T_2  \mathbf {\tilde y}_l$ and $\bm \Sigma_{\tilde b_{2,l}}= \sigma^2_0(\mathbf H^T_2 \mathbf H_2)^{-1}$ for $l=d+1,...,n_1$.

	
\end{enumerate}
\label{lemma:marginal_post_B_1_B_2}
\end{lemma}

\begin{proof}[Proof of Lemma \ref{lemma:marginal_post_B_1_B_2}] 

\begin{enumerate}
\item (Marginal distribution of transformed row regression coefficients). 

Denote $(\mathbf{B}^{aug}_1)=[\mathbf B_1^T, \mathbf{\tilde B}_{1,(q_1+1):n_1}]^T$, where $\mathbf{\tilde B}_{1,(q_1+1):n_1}$ are the last $n_1-q_1$ columns of $\mathbf{\tilde B}_1$. Denote $p_{trans}(\mathbf{B}_1  \mid \mathbf{Y},\bm \Theta_{-\mathbf{B}_1})$ and $p_{trans}(\mathbf{B}^{aug}_1  \mid \mathbf{Y},\bm \Theta_{-\mathbf{B}_1})$ the transformed marginal posterior distribution of $\mathbf{B}_1$ and $\mathbf {B}^{aug}_1$ derived by transforming 	$p(\mathbf { \tilde B}_1 \mid \mathbf Y, \bm \Theta_{-B_1})$ in (\ref{equ:tilde_B1}). We have
\begin{align*}
&p_{trans}(\mathbf{B}_1  \mid \mathbf{Y},\bm \Theta_{-\mathbf{B}_1}) \propto p_{trans}(\mathbf{B}^{aug}_1  \mid \mathbf{Y},\bm \Theta_{-\mathbf{B}_1})=p(\mathbf{\tilde B}_1  \mid \mathbf{Y},\bm \Theta_{-\mathbf{B}_1}) \left| \frac{d \mathbf{\tilde B}_1}{ d \mathbf{B}^{aug}_1 }\right|\\
\propto& \exp\left\{-\frac{1}{2}\sum_{l=1}^d (\mathbf{\tilde b}_{1,l}-\mathbf{\tilde{y}}_l )^T\bm {\tilde \Sigma}^{-1}_l(\mathbf{\tilde b}_{1,l}-\mathbf{\tilde{y}}_l )-\frac{1}{2\sigma_0^2}\sum_{l=d+1}^{n_1}(\bm {\tilde b}_{1,l} -\mathbf{\tilde{y}}_l)^T(\bm {\tilde b}_{1,l}-\mathbf{\tilde{y}}_l) \right\} \\
\propto & \exp\left\{-\frac{1}{2}\sum_{l=1}^d \mathbf{a}_l^T(\mathbf{Y}-\mathbf{H}_1 \mathbf{B}_1)\bm {\tilde \Sigma}^{-1}_l(\mathbf{Y}-\mathbf{H}_1\mathbf{B}_1)^T\mathbf{a}_l \right.\\
&\hspace{1in} \left.-\frac{1}{2\sigma_0^2}\sum_{l=d+1}^{n_1}\mathbf{a}_l^T(\mathbf{Y}-\mathbf{H}_1 \mathbf{B}_1)(\mathbf{Y}-\mathbf{H}_1 \mathbf{B}_1)^T\mathbf{a}_l\right\},
\end{align*}
where the last line is the same as the posterior distribution of $\mathbf B_1$ based on the marginal likelihood in equation (\ref{equ:marginal_lik}) and the prior distribution $\pi(\mathbf{B}_1) \propto 1$. 
Thus if one sample $\mathbf {\tilde B_1}$ from (\ref{equ:tilde_B1}), one can obtain the sample for $(\mathbf{B}_1 \mid \mathbf{Y},\bm \Theta_{-\mathbf{B}_1}) $ through $(\mathbf H^T_1 \mathbf H_1)^{-1}\mathbf H^T_1  \mathbf A_F \mathbf {\tilde B_1}^T$.

\item (Marginal distribution of transformed column regression coefficients). 

Since $\pi(\mathbf{B}_2) \propto 1$ is a Jeffreys prior,   and $\mathbf{\tilde{B}}_2$ is a linear transformation of $\mathbf B_2$ with the same dimension, we have $\pi(\mathbf{\tilde{B}}_2) \propto 1$.

Based on the marginal likelihood in  equation (\ref{equ:marginal_lik}) and the prior distribution, the posterior distribution of $\mathbf{\tilde{B}}_2$ follows:
\begin{align*}
& p(\mathbf{\tilde{B}}_2 \mid \mathbf{Y},\bm \Theta_{-\mathbf{B}_2})  \\
\propto & \exp\left\{-\frac{1}{2}\sum_{l=1}^d \mathbf{a}_l^T(\mathbf{Y}-\mathbf{B}_2^T \mathbf{H}_2^T)\bm {\tilde \Sigma}^{-1}_l(\mathbf{Y}-\mathbf{B}_2^T \mathbf{H}_2^T)^T\mathbf{a}_l \right.\\
&  \hspace{1in} \left. -\frac{1}{2\sigma_0^2}\sum_{l=d+1}^{n_1}\mathbf{a}_l^T(\mathbf{Y}-\mathbf{B}_2^T \mathbf{H}_2^T)(\mathbf{Y}-\mathbf{B}_2^T \mathbf{H}_2^T)^T\mathbf{a}_l\right\} \\
\propto& \exp\left\{-\frac{1}{2}\sum_{l=1}^d (\mathbf{\tilde{y}}_l-\mathbf{H}_2 \mathbf {\tilde b}_{2,l} )^T\bm {\tilde \Sigma}^{-1}_l(\mathbf{\tilde{y}}_l-\mathbf{H}_2 \mathbf {\tilde b}_{2,l} ) \right. \\
& \hspace{1in} \left. -\frac{1}{2\sigma_0^2}\sum_{l=d+1}^{n_1}(\mathbf{\tilde{y}}_l-\mathbf{H}_2 \mathbf {\tilde b}_{2,l})^T(\mathbf{\tilde{y}}_l-\mathbf{H}_2 \mathbf {\tilde b}_{2,l} )\right\} \\
\propto& \exp\left\{-\frac{1}{2}\sum_{l=1}^d \left(\mathbf {\tilde b}_{2,l} -\bm \mu_{\tilde b_{2,l}}\right)^T \mathbf H^T_2 \bm {\tilde \Sigma}^{-1}_l  \mathbf H_2 \left(\mathbf {\tilde b}_{2,l} -\bm \mu_{\tilde b_{2,l}}\right)\right. \\ 
& \left.\hspace{1in}-\frac{1}{2\sigma_0^2}\sum_{l=d+1}^{n_1}\left(\mathbf {\tilde b}_{2,l} -\bm \mu_{\tilde b_{2,l}}\right)^T \mathbf{H}_2^T \mathbf{H}_2 \left(\mathbf {\tilde b}_{2,l}-\bm \mu_{\tilde b_{2,l}}\right) \right\},
\end{align*}
from which equation (\ref{equ:tilde_B2}) follows.

\end{enumerate}

\end{proof}

We are ready to prove Theorem \ref{thm:marginal_post_B_1_B_2}. 

\begin{proof}[Proof of Theorem \ref{thm:marginal_post_B_1_B_2}]
After marginalizing out $\mathbf{Z}$, we have

\begin{enumerate}
\item (Row regression coefficients). 

From Lemma \ref{lemma:marginal_post_B_1_B_2},   the posterior mean of $	(\mathbf { \tilde B}_1 \mid \mathbf Y, \bm \Theta_{-\mathbf{B}_1})$ is $\mathbf Y^T \mathbf A_F$, where $\mathbf{A}_F := [\mathbf{A}_s, \mathbf{A}_c]$.
 We denote the centered $\mathbf{\tilde B}_{1}$ by $\mathbf{\tilde B}_{1,0}= [\mathbf{\tilde B}_{1,0,s}, \mathbf{\tilde B}_{1,0,c}] = \mathbf{\tilde B}_1 - \mathbf Y^T \mathbf A_F$, where $\mathbf{\tilde B}_{1,0,s}$ is the first $d$ columns of $\mathbf{\tilde B}_{1,0}$ and $\mathbf{\tilde B}_{1,0,c}$ is the last $(n_1-d)$ columns of $\mathbf{\tilde B}_{1,0}$. Let $\tilde{{\mathbf{b}}}_{1,0,l}$ be the $l$-th column of $\mathbf{\tilde B}_{1,0}$. Then the posterior mean of $(\mathbf{B}_1 \mid \mathbf{Y}, \bm{\Theta}_{-\mathbf{B}_1})$ can be calculated below
\begin{align*}
    \mathbf{\hat{B}}_1  &  =  \E ( \mathbf{B}_1 \mid \mathbf{Y}, \bm{\Theta}_{-\mathbf{B}_1}) = \E\left((\mathbf{H}_1^T\mathbf{H}_1)^{-1} \mathbf{H}_1^T \mathbf{A}_F \mathbf{\tilde B}_1^T \mid \mathbf{Y}, \bm{\Theta_{-\mathbf{B}_1}} \right) \\
     & = (\mathbf{H}_1^T\mathbf{H}_1)^{-1} \mathbf{H}_1^T \mathbf{A}_F \mathbf{A}_F^T \mathbf{Y} = (\mathbf{H}_1^T\mathbf{H}_1)^{-1} \mathbf{H}_1^T \mathbf{Y}
\end{align*}


Note $\mathbf B_1=(\mathbf H^T_1\mathbf H_1)^{-1}\mathbf H_1^T \mathbf A_F \mathbf{\tilde B}_1^T$, one has 
\begin{align*}
    \mathbf{B}_1 - \mathbf{\hat{B}}_1 & = (\mathbf H^T_1\mathbf H_1)^{-1} \mathbf H_1^T \mathbf A_F (\mathbf { \tilde B}_{1,0})^T = (\mathbf H^T_1\mathbf H_1)^{-1} \mathbf H_1^T \left( \mathbf{A}_s \mathbf{\tilde B}_{1,0,s}^T + \mathbf{A}_c \mathbf{\tilde B}_{1,0,c}^T\right)
\end{align*}
where $ \mathbf {\tilde B}_{1,0,s}$ is a $n_2 \times d$ matrix with the $lth$ column independently sampled from $\mathcal N(\mathbf 0, \bm {\tilde \Sigma}_l)$ for $l=1,...,d$. For the distribution of $\mathbf{A}_c \mathbf{\tilde B}^T_{1,0,c}$, using part 1 of Lemma \ref{lemma:marginal_post_B_1_B_2}, we have
\begin{align*}
p(\mathbf{A}_c\mathbf{\tilde B}^T_{1,0,c}\mid \mathbf{Y}, \bm{\Theta}_{-\mathbf{B}_1})&\propto \exp\left\{-\frac{1}{2\sigma_0^2} \tr\left(\mathbf A^T_c \mathbf{\tilde B}_{1,0,c}\mathbf{\tilde B}^T_{1,0,c}\mathbf A_c  \right)\right\} \\
&\propto \exp\left\{-\frac{1}{2\sigma_0^2}\tr\left((\mathbf{I}_{n_1} - \mathbf{A}_s \mathbf{A}_s^T) \mathbf{\tilde B}_{1,0,c} \mathbf{\tilde B}_{1,0,c}^T\right)\right\}.
\end{align*}
Thus we can sample $\mathbf{A}_c \mathbf{\tilde B}^T_{1,0,c}$ by  $\sigma_0 (\mathbf{I}_{n_1} - \mathbf{A}_s\mathbf{A}_s^T)\mathbf{Z}_{0,1}$, where $\mathbf Z_{0,1}$ is an $n_1 \times n_2$ matrix with each entry independently sampled from standard normal distribution. The results soon follow.

\item (Column regression coefficients). 


We first compute the posterior mean of $(\mathbf{B}_2\mid \mathbf{Y}, \bm{\Theta}_{-\mathbf{B}_2})$ below
\begin{align*}
    \mathbf{\hat{B}}_2 & = \E(\mathbf{B}_2 \mid \mathbf{Y}, \bm{\Theta}_{-\mathbf{B}_2}) = \E( \mathbf{\tilde B}_{2} \mathbf{A}_F^T \mid \mathbf{Y}, \bm{\Theta_{-\mathbf{B}_2}})  \\
    & = \sum^d_{l=1} (\mathbf H^T_2 \bm{\tilde \Sigma}_l^{-1} \mathbf H_2)^{-1} \mathbf H^T_2 \bm{\tilde \Sigma}^{-1}_l \mathbf Y^T\mathbf a_l \mathbf a^T_l + (\mathbf H^T_2 \mathbf H_2)^{-1} \mathbf H^T_2 \mathbf Y^T (\mathbf I_{n_1}-\mathbf A_s \mathbf A^T_s) 
\end{align*}



We denote the centered $\mathbf{\tilde B}_{2}$ by
$\mathbf{\tilde B}_{2,0}=[ \mathbf{\tilde B}_{2,0,s}, \mathbf{\tilde B}_{2,0,c}]$. We have 
\begin{align*}
    \mathbf{B}_2 - \mathbf{\hat{B}}_2 = & \mathbf { \tilde B}_{2,0} \mathbf{A}_F^T =  \mathbf{\tilde B}_{2,0,s} \mathbf{A}_s^T + \mathbf{\tilde B}_{2,0,c} \mathbf{A}_c^T
    \end{align*}
where $\mathbf {\tilde B}_{2,0,s}$ is a $q_2\times d$ matrix with the $l$th column independently sampled from $\mathcal N(\mathbf 0, (\mathbf H^T_2 \bm {\tilde \Sigma}^{-1}_l  \mathbf H_2)^{-1})$ for $l=1,...,d$. For the distribution of  $\mathbf{\tilde B}_{2,0,c} \mathbf{A}_c^T$, we have
\begin{align*}
p(\mathbf{\tilde B}_{2,0,c} \mathbf{A}_c^T\mid \mathbf{Y}, \bm{\Theta}_{-\mathbf{B}_2})\propto  & exp\left\{-\frac{1}{2\sigma_0^2}\tr\left(\mathbf{A}_c \mathbf{A}_c^T \mathbf{\tilde  B}_{2,0,c} \mathbf{H}_2^T \mathbf{H}_2 \mathbf{\tilde  B}_{2,0,c}^T\right)\right\} \\
    \propto & exp\left\{-\frac{1}{2\sigma_0^2}\tr\left((\mathbf{I}_{n_1} - \mathbf{A}_s \mathbf{A}_s^T) \mathbf{\tilde B}_{2,0,c} \mathbf{H}_2^T \mathbf{H}_2 \mathbf{\tilde  B}_{2,0,c}^T\right)\right\}.
\end{align*}
Thus we can sample $\mathbf{\tilde B}_{2,0,c} \mathbf{A}_c^T$ by $\sigma_0 (\mathbf{I}_{n_1} - \mathbf{A}_s\mathbf{A}_s^T)\mathbf{Z}_{0,2}^T \mathbf{L}_{H_2}^T$, where $\mathbf L_{H_2}$ is a $q_2 \times q_2$ matrix such that  $\mathbf L_{H_2}\mathbf L^T_{H_2}=(\mathbf H^T_2 \mathbf H_2)^{-1}$ and $\mathbf Z_{0,2}$ is a $q_2 \times n_1$ matrix with each entry independently sampled from standard normal distribution.

\end{enumerate}

\end{proof}

\begin{lemma}
Assume $\mathbf M=\mathbf H_1 \mathbf B_1 + (\mathbf H_2 \mathbf B_2)^T$ and let the objective prior $\pi(\mathbf B_1, \mathbf B_2)\propto 1$ for the regression parameters $\mathbf B_1$ and $\mathbf B_2$. Denote $  \mathbf { \tilde B}_1=[ \tilde {\mathbf{b}}_{1,1},..., \tilde {\mathbf{b}}_{1,n_1}]= \mathbf B_1^T \mathbf H_1^T \mathbf A_{F}$ and $\mathbf { \tilde B}_2=[ \tilde {\mathbf{b}}_{2,1},..., \tilde {\mathbf{b}}_{2,n_1}]=\mathbf B_2 \mathbf A_{F}$. 
\begin{enumerate}
\item After marginalizing out $\mathbf Z$ and $\mathbf B_1$, assume the marginal posterior distribution of $\mathbf { \tilde B}_1$ follows
	\begin{equation}
	p(\mathbf { \tilde B}_1 \mid \mathbf Y, \bm \Theta_{-\mathbf{B}_1, -\mathbf{B}_2})= \prod^{n_1}_{l=1} \mathcal{PN}(\bm {\tilde b}_{1,l}; \mathbf {\tilde y}_l,  \mathbf Q_{1,l} ).  
	\label{equ:tilde_B1_marginalize_B2}
	\end{equation}
where $\mathbf Q_{1,l}=\mathbf P_l^T (\bm {\tilde \Sigma}_l)^{-1} \mathbf P_l$, $\mathbf P_l= \mathbf I_{n_2} - \mathbf{H}_2(\mathbf{H}_2^T \bm {\tilde \Sigma}_l^{-1} \mathbf{H}_2)^{-1}\mathbf{H}_2^T  \bm {\tilde \Sigma}_l^{-1}$ for $l=1,...,d$ and    $\mathbf Q_{1,l}=\sigma^2_0 \mathbf P_0$ with $\mathbf P_0=(\mathbf{I}_{n_2} - \mathbf{H}_2(\mathbf{H}_2^T\mathbf{H}_2)^{-1}\mathbf{H}_2^T)$ for $l=d+1,...,n_1$. The sample $(\mathbf B_1 \mid \mathbf Y, \bm \Theta_{-\mathbf{B}_1, -\mathbf{B}_2})$ can be obtained by $(\mathbf H^T_1 \mathbf H_1)^{-1}\mathbf H^T_1  \mathbf A_F \mathbf {\tilde B}_1^T$, where  $\mathbf {\tilde B}_1$  sampled from the $p(\mathbf { \tilde B}_1 \mid \mathbf Y, \bm \Theta_{-\mathbf{B}_1, -\mathbf{B}_2})$ in equation (\ref{equ:tilde_B1_marginalize_B2}). 


\item After marginalizing out $\mathbf Z$ and conditional on $\mathbf B_1$, the marginal posterior distribution of $\mathbf { \tilde B}_2$ follows (\ref{equ:tilde_B2}) by replacing $\mathbf {\tilde y}_l$ by $\mathbf {\tilde y}_{l,B_1}= (\mathbf Y- \mathbf H_1 \mathbf B_1)^T \mathbf a_l$ for $l=1,...,d$. 
\end{enumerate}

\label{lemma:sample_B12}
\end{lemma}

\begin{proof}[Proof of Lemma \ref{lemma:sample_B12}]
Denote $\mathbf{Y}_0 = \mathbf{Y} - \mathbf{H}_1 \mathbf{B}_1 - \mathbf{B}_2^T \mathbf{H}_2^T$. Define $\mathbf{G} = [\mathbf{g}_1, \mathbf{g}_2,...,\mathbf{g}_{n_1}]= (\mathbf{Y} - \mathbf{H}_1 \mathbf{B}_1)^T \mathbf{A}_F$. That is, $\mathbf{g}_l=  (\mathbf{Y} - \mathbf{H}_1 \mathbf{B}_1)^T \mathbf{a}_l$.

First we have the joint posterior distribution $(\mathbf{ B}_1, \mathbf{ B}_2 \mid \mathbf{Y}, \bm{\Theta}_{-\mathbf{B}_1, -\mathbf{B}_2})$
\begin{align*}
    &p(\mathbf{ B}_1, \mathbf{ B}_2 \mid \mathbf{Y}, \bm{\Theta}_{-\mathbf{B}_1, -\mathbf{B}_2}) \\
    \propto & \exp\left\{ -\frac{1}{2} \sum_{l=1}^d \mathbf{a}_l^T \mathbf{Y}_0^T \bm{\tilde \Sigma}_l^{-1} \mathbf{Y}_0 \mathbf{a}_l - \frac{1}{2\sigma_0^2} \sum_{l=d+1}^{n_1} \mathbf{a}_l^T \mathbf{Y}_0^T \mathbf{Y}_0 \mathbf{a}_l  \right\} \\
    \propto & \exp\left\{ -\frac{1}{2} \sum_{l=1}^d (\mathbf{g}_l - \mathbf{H}_2 \mathbf{\tilde b}_{2,l})^T \bm{\tilde \Sigma}_l^{-1} (\mathbf{g}_l - \mathbf{H}_2 \mathbf{\tilde b}_{2,l}) - \frac{1}{2\sigma_0^2} \sum_{l=d+1}^{n_1} (\mathbf{g}_l - \mathbf{H}_2 \mathbf{\tilde b}_{2,l})^T (\mathbf{g}_l - \mathbf{H}_2 \mathbf{\tilde b}_{2,l})  \right\},
\end{align*}
where $\mathbf{\tilde b}_{2,l}$ is a transformation of $\mathbf B_2$ defined in part 2 in Lemma \ref{lemma:marginal_post_B_1_B_2}.

After integrating out $\mathbf{\tilde b}_{2,l}$ from $p(\mathbf{ B}_1, \mathbf{ B}_2 \mid \mathbf{Y}, \bm{\Theta}_{-\mathbf{B}_1, -\mathbf{B}_2})$ for $l=1,2...,n_1$, one has
\begin{align*}
    & p(\mathbf{ B}_1 \mid \mathbf{Y}, \bm{\Theta}_{-\mathbf{B}_1, -\mathbf{B}_2}) \\
    \propto & \exp\left\{-\frac{\sum_{l=1}^{d}(\mathbf{g}_l -\mathbf{H}_2 \hat{\mathbf{b}}_{2,l} )^T \tilde{\bm \Sigma}_l^{-1} (\mathbf{g}_l -\mathbf{H}_2 \hat{\mathbf{b}}_{2,l} ) }{2}
    -\frac{\sum_{l=d+1}^{n_1}(\mathbf{g}_l -\mathbf{H}_2 \hat{\mathbf{b}}_{2,l} )^T (\mathbf{g}_l -\mathbf{H}_2 \hat{\mathbf{b}}_{2,l} )}{2 \sigma_0^2}\right\}\\
    \propto  & \exp\left\{-\frac{\sum_{l=1}^{d} \mathbf{g}_l^T \mathbf{P}_l^T (\tilde{\bm{\Sigma}}_l)^{-1} \mathbf P_l \mathbf{g}_l}{2} -\frac{\sum_{l=d+1}^{n_1} \mathbf{g}_l^T \mathbf{P}_{0} \mathbf{g}_l}{2 \sigma_0^2}\right\} \\
    \propto & \exp\left\{-\frac{\sum_{l=1}^{n_1} \mathbf{g}_l^T \mathbf Q_{1,l} \mathbf{g}_l}{2}\right\}
\end{align*}
Where
\begin{equation*}
\hat{\mathbf{b}}_{2,l} = \left\{
\begin{aligned}
& (\mathbf{H}_2^T \bm{\tilde{\Sigma}}_l^{-1} \mathbf{H}_2)^{-1} \mathbf{H}_c^T \bm{\tilde{\Sigma}}_l^{-1}\mathbf{g}_l & & l=1,2,...,d \\
& (\mathbf H^T_2 \mathbf H_2)^{-1}\mathbf H^T_2\mathbf {g}_l & & l = d+1,...,n_1\\
\end{aligned}
\right.
\end{equation*}

Denote $\mathbf{B}^{aug}_1=[\mathbf B_1^T, \mathbf{\tilde B}_{1,(q_1+1):n_1}]^T$, where $\mathbf{\tilde B}_{1,(q_1+1):n_1}$ is the last $n_1-q_1$ columns of $\mathbf{\tilde B}_1$. Denote  the marginal posterior distribution $p_{trans}(\mathbf{B}_1^T \mid \mathbf{Y}, \bm{\Theta}_{-\mathbf{B}_1, -\mathbf{B}_2})$ and 
$p_{trans}(\mathbf{B}^{aug}_1 \mid \mathbf{Y}, \bm{\Theta}_{-\mathbf{B}_1, -\mathbf{B}_2})$ derived by the transformation of $p(\mathbf{\tilde B}_1 \mid \mathbf{Y}, \bm{\Theta}_{-\mathbf{B}_1, -\mathbf{B}_2})$ . One has
\begin{align*}
p_{trans}(\mathbf{B}_1^T \mid \mathbf{Y}, \bm{\Theta}_{-\mathbf{B}_1, -\mathbf{B}_2}) \propto& p(\mathbf{B}^{aug}_1  \mid \mathbf{Y},\Theta_{-\mathbf{B}_1, -\mathbf{B}_2})\\
=&p(\mathbf{\tilde B}_1  \mid \mathbf{Y},\Theta_{-\mathbf{B}_1, -\mathbf{B}_2}) \left| \frac{d \mathbf{\tilde B}_1}{ d \mathbf{B}^{aug}_1 }\right| \\
     \propto&  \exp\left\{-\frac{\sum_{l=1}^{n_1} \mathbf{g}_l^T \mathbf Q_{1,l} \mathbf{g}_l}{2}\right\}
\end{align*}

Because $\mathbf{Q}_{1,l}$ is idempotent, i.e. $\mathbf{Q}_{1,l}\mathbf{Q}_{1,l}=\mathbf{Q}_{1,l}$, the Moore–Penrose inverse of $\mathbf{Q}_{1,l}$ is $\mathbf{Q}_{1,l}$ itself. Therefore for $l=1,...,d$, $ \tilde {\mathbf{b}}_{1,l}\mid \mathbf Y, \bm \Theta_{-\mathbf{B}_1, -\mathbf{B}_2}\sim \mathcal M( \mathbf {\tilde y}_l, \, \mathbf{Q}_{1,l}  )$, from which the part 1 follows. Part 2 follows Lemma \ref{lemma:marginal_post_B_1_B_2}.

\end{proof}

We are ready to prove Theorem \ref{thm:sample_B12}.

\begin{proof}[Proof of Theorem \ref{thm:sample_B12}]

By Lemma \ref{lemma:sample_B12}, the posterior mean of $\mathbf{\tilde B}_1\mid \mathbf{Y},\bm \Theta_{-B_1,-B_2}$ is $\mathbf{Y}^T \mathbf{A}_F$, where $\mathbf{A}_F := [\mathbf{A}_s, \mathbf{A}_c]$. We denote the centered $\mathbf{\tilde B}_1$ by $\mathbf{\tilde B}_{1,0} = [\mathbf{\tilde B}_{1,Q}, \mathbf{\tilde B}_{1,0,c}] = \mathbf{\tilde B}_1 - \mathbf{Y}^T \mathbf{A}_F$, where $\mathbf{\tilde B}_{1,Q}$ is the first $d$ columns of $\mathbf{\tilde B}_{1,0}$ and $\mathbf{\tilde B}_{1,0,c}$ is the next $(n_1-d)$ columns of $\mathbf{\tilde B}_{1,0}$. Then the posterior mean of $\mathbf{B}_1\mid \mathbf{Y}, \bm{\Theta}_{-B_1,-B_2}$ can be calculated below

\begin{align*}
    \mathbf{\hat{B}}_1 & = \E\left(\mathbf{B}_1\mid \mathbf{Y}, \bm{\Theta}_{-B_1,-B_2}\right) = \E\left((\mathbf{H}_1^T\mathbf{H}_1)^{-1} \mathbf{H}_1^T \mathbf{A}_F \mathbf{\tilde B}_1^T \mid \mathbf{Y}, \bm{\Theta}_{-B_1,-B_2}\right) \\
    & =(\mathbf{H}_1^T\mathbf{H}_1)^{-1} \mathbf{H}_1^T \mathbf{A}_F \mathbf{A}_F^T \mathbf{Y} = (\mathbf{H}_1^T\mathbf{H}_1)^{-1} \mathbf{H}_1^T \mathbf{Y}
\end{align*}




Note $\mathbf B_1=(\mathbf H^T_1\mathbf H_1)^{-1}\mathbf H_1^T \mathbf A_F \mathbf{\tilde B}_1^T$, one has 

\begin{align*}
     \mathbf{B}_1 - \mathbf{\hat{B}}_1 &= (\mathbf H^T_1\mathbf H_1)^{-1} \mathbf H_1^T \mathbf A_F (\mathbf { \tilde B}_{1,0})^T = (\mathbf H^T_1\mathbf H_1)^{-1} \mathbf H_1^T ( \mathbf{A}_s (\mathbf{\tilde B}_{1,Q})^T + \mathbf{A}_c (\mathbf{\tilde B}_{1,0,c})^T)
\end{align*}

where by Lemma \ref{lemma:sample_B12}, $ \mathbf {\tilde B}_{1,Q}$ is an $n_2 \times d$ matrix with the $lth$ column independently sampled from $\mathcal N(\mathbf 0, \mathbf{Q}_{1,l})$ for $l=1,...,d$.
For the distribution of $\mathbf{A}_c \mathbf{\tilde B}^T_{1,0,c}$, using part 1 of Lemma \ref{lemma:sample_B12}, we have 
\begin{align*}
      &p(\mathbf{A}_c \mathbf{\tilde B}_{1,0,c}^T | \mathbf{Y}, \bm{\Theta}_{-B_1,-B_2})   \\
    \propto  & \exp\left\{\frac{1}{2\sigma_0^2}\tr\left(\mathbf{A}_c \mathbf{A}_c^T \mathbf{\tilde B}_{1,0,c} \mathbf{P}_{0} (\mathbf{\tilde B}_{1,0,c})^T\right)\right\} \\
    \propto & \exp\left\{-\frac{1}{2\sigma_0^2}\tr\left((\mathbf{I}_{n_1} - \mathbf{A}_s \mathbf{A}_s^T) \mathbf{\tilde B}_{1,0,c} \mathbf{P}_{0} (\mathbf{\tilde B}_{1,0,c})^T\right)\right\}.
\end{align*}
Thus we can sample marginal posterior distribution of $\mathbf{A}_c \mathbf{\tilde B}_{1,0,c}^T$ by $\sigma_0 (\mathbf{I}_{n_1} - \mathbf{A}_s\mathbf{A}_s^T)\mathbf{Z}_{0,1} \mathbf{P}_{0}$, where $\mathbf Z_{0,1}$ is an $n_1 \times n_2$ matrix with each entry independently sampled from standard normal distribution. The results soon follow.

\end{proof}
  \begin{figure}[t]
\centering
  \begin{tabular}{ccc}
  \hspace{-.15in}
	\includegraphics[height=.4\textwidth,width=.34\textwidth]{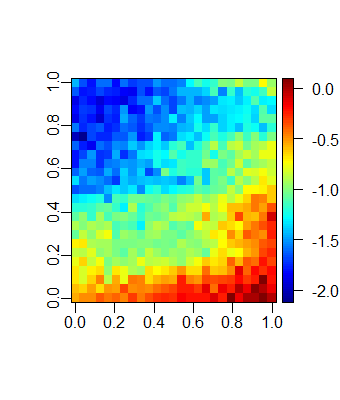} \hspace{-.1in}
	\includegraphics[height=.4\textwidth,width=.34\textwidth]{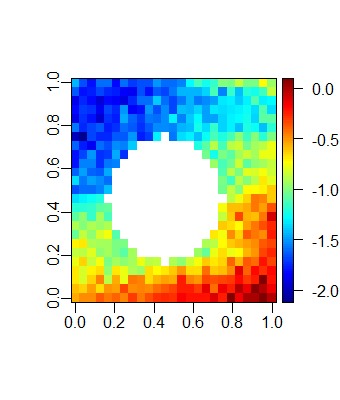} 
	\hspace{-.1in}
	\includegraphics[height=.4\textwidth,width=.34\textwidth]{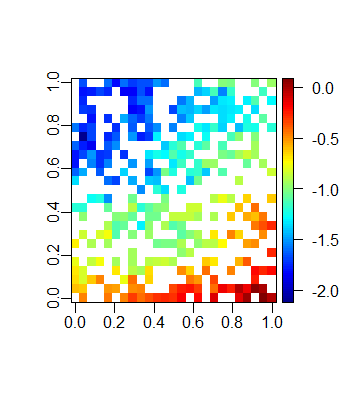}
  \end{tabular}
  \vspace{-.4in}
   \caption{The simulated data  with  full observations, disk missing pattern and missing-at-random pattern with $50\%$ of the missing values are graphed in the left, middle and right panels, respectively.}
\label{fig:missing_pattern}
\end{figure}

\section{Additional results of  simulated studies in Section \ref{sec:simulation} }
\label{sec:supp_simulation}

\begin{figure}[t]
\includegraphics[height=.7\textwidth,width=.8\textwidth]{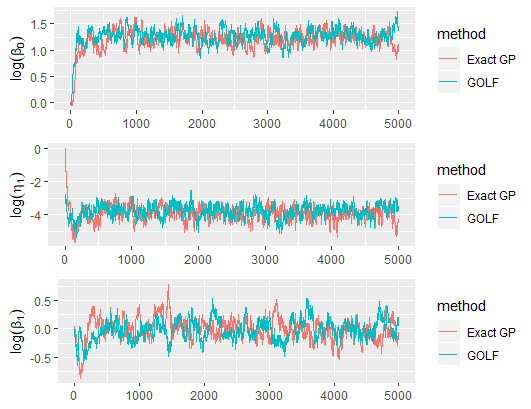}
\centering
   \caption{Trace plots of posterior samples of  parameters in the exact GP model and GOLF processes for the simulated data in figure \ref{fig:missing_pattern}.}
\label{fig:traceplots}
\end{figure}

We provide additional results for the simulated studies in Example \ref{eg:compare_exact} in Figure \ref{fig:missing_pattern} and Figure \ref{fig:traceplots}. We graph the simulated data set with  full observations, disk missing pattern and missing-at-random pattern with $50\%$ of the missing values in Figure \ref{fig:missing_pattern}.  Posterior samples of the logarithm of the inverse range parameter of factor loading matrix,   the nugget parameter and  the inverse range parameter of the factors are graphed from the upper to lower panels in Figure \ref{fig:traceplots}, respectively. The posterior samples of parameters in the exact GP model and GOLF processes are similar to each other.




\begin{figure}[t]
\centering
  \begin{tabular}{ccc}
	\includegraphics[height=.3\textwidth,width=.33\textwidth]{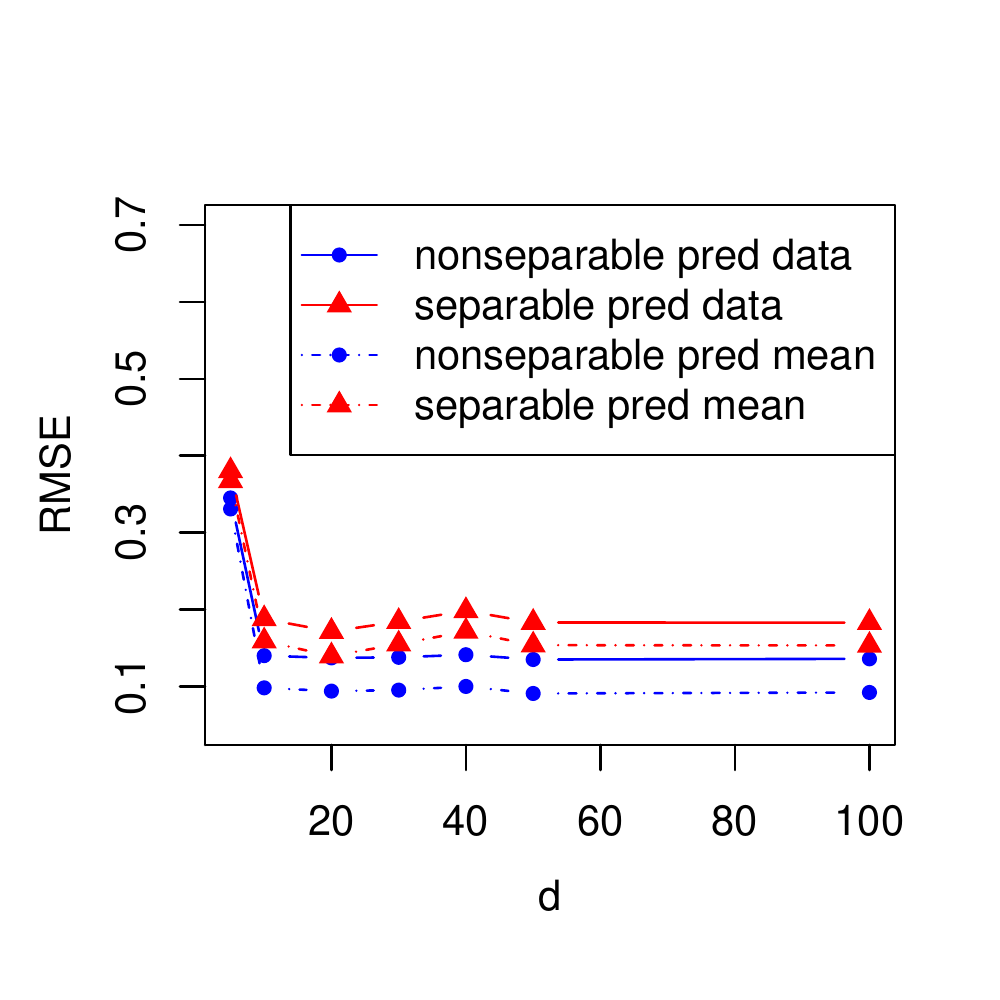}
	\includegraphics[height=.3\textwidth,width=.33\textwidth]{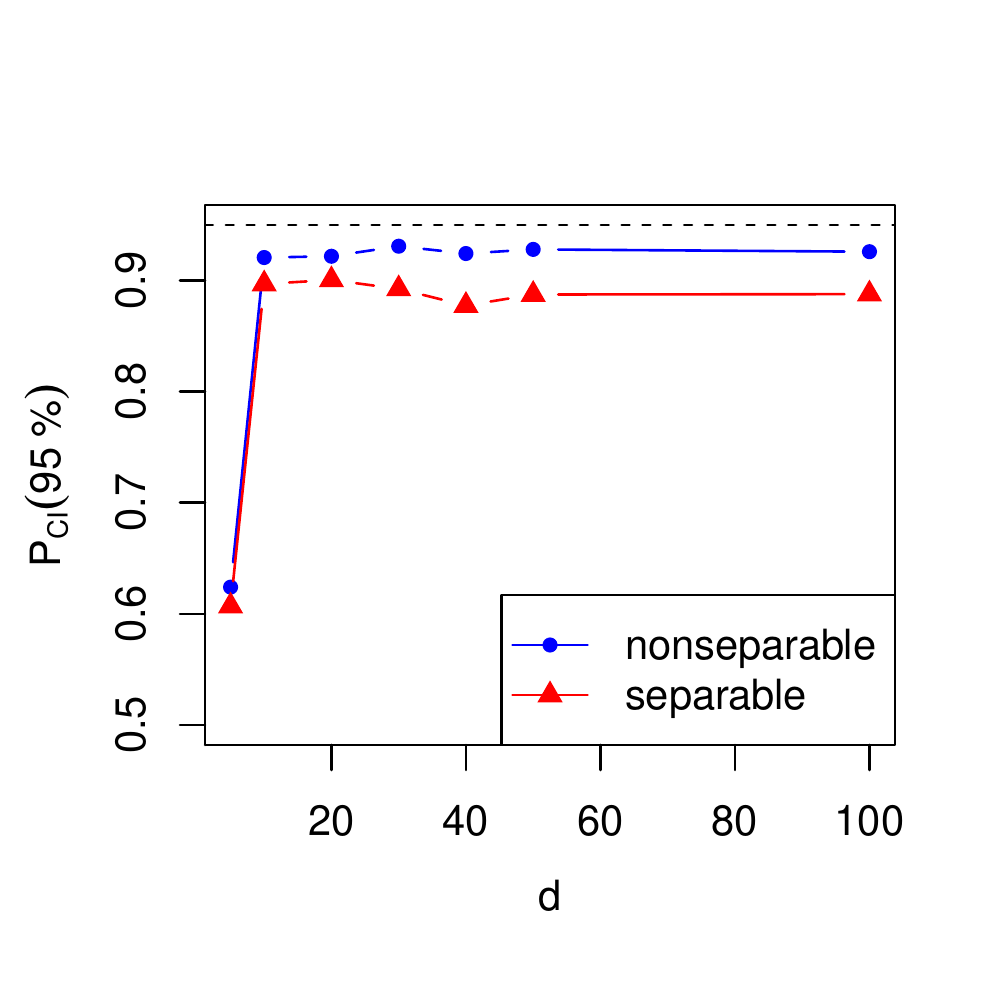} 
	\includegraphics[height=.3\textwidth,width=.33\textwidth]{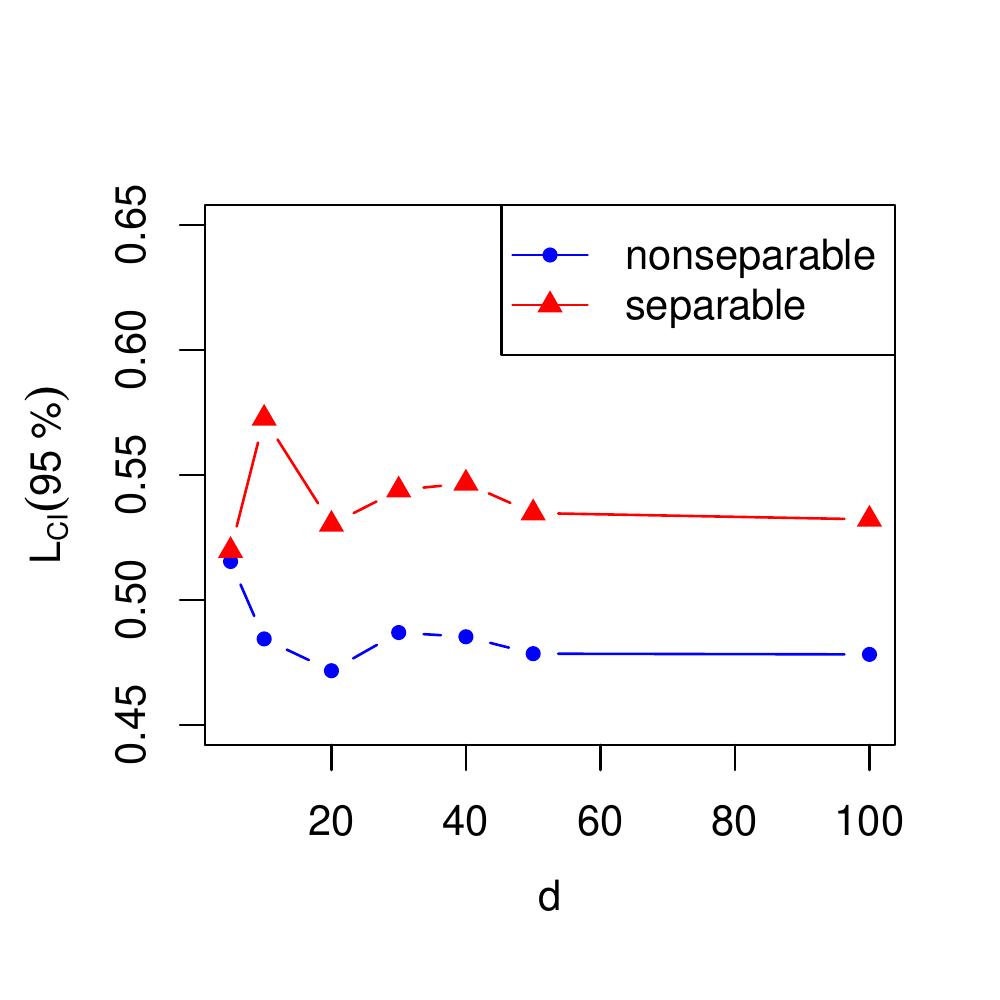}   \vspace{-.4in}
 \\
		\includegraphics[height=.3\textwidth,width=.33\textwidth]{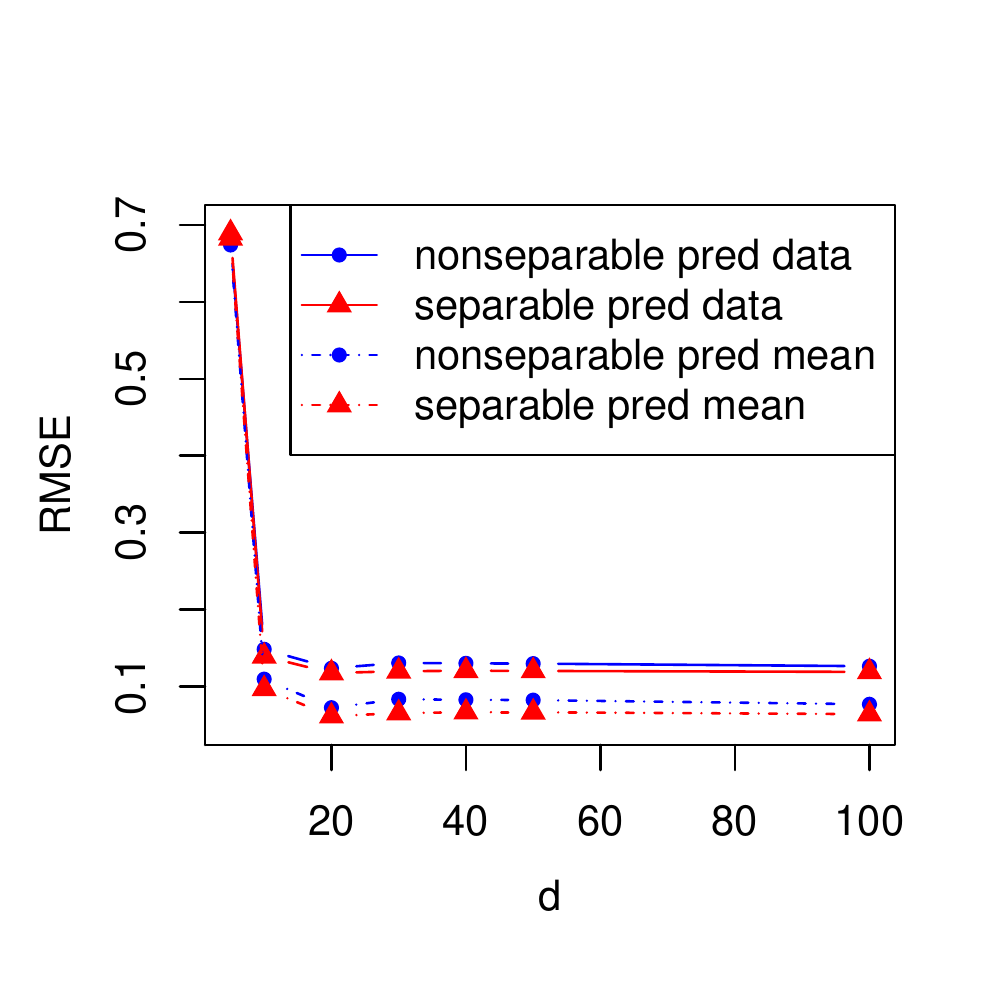}
	\includegraphics[height=.3\textwidth,width=.33\textwidth]{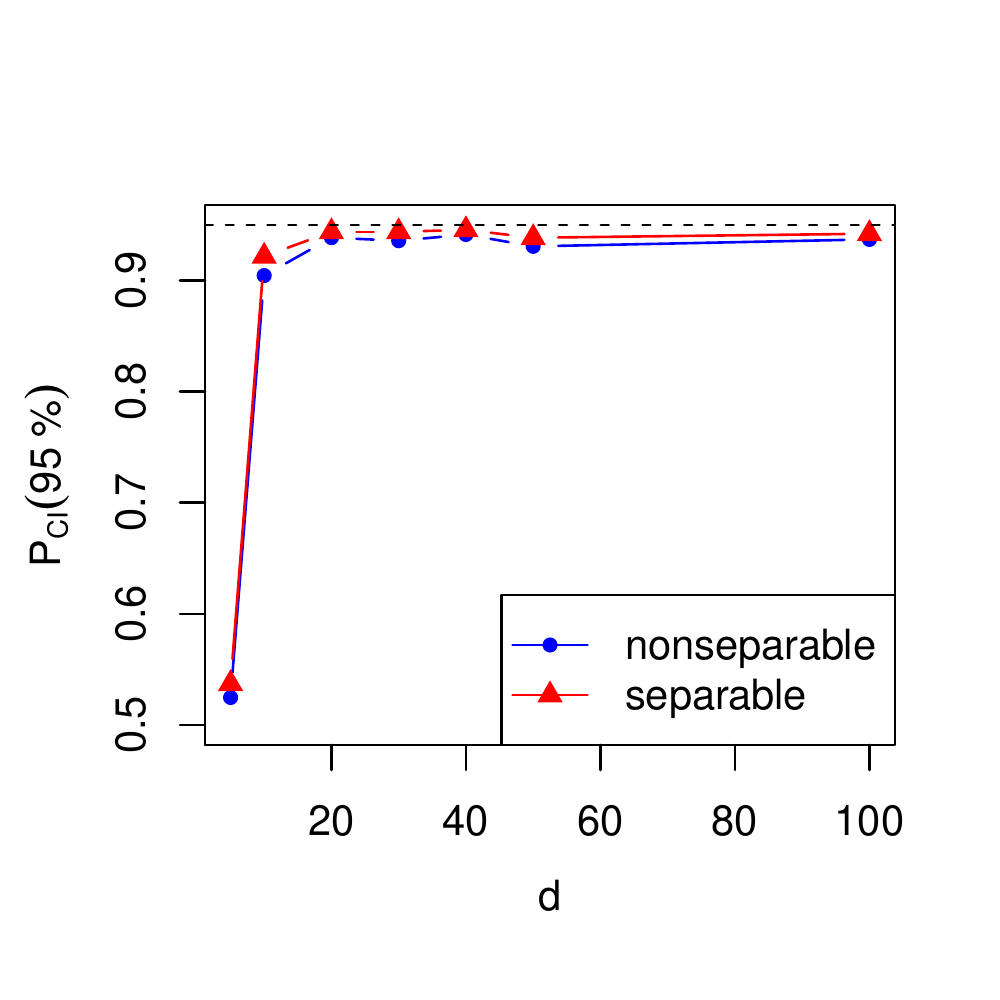} 
	\includegraphics[height=.3\textwidth,width=.33\textwidth]{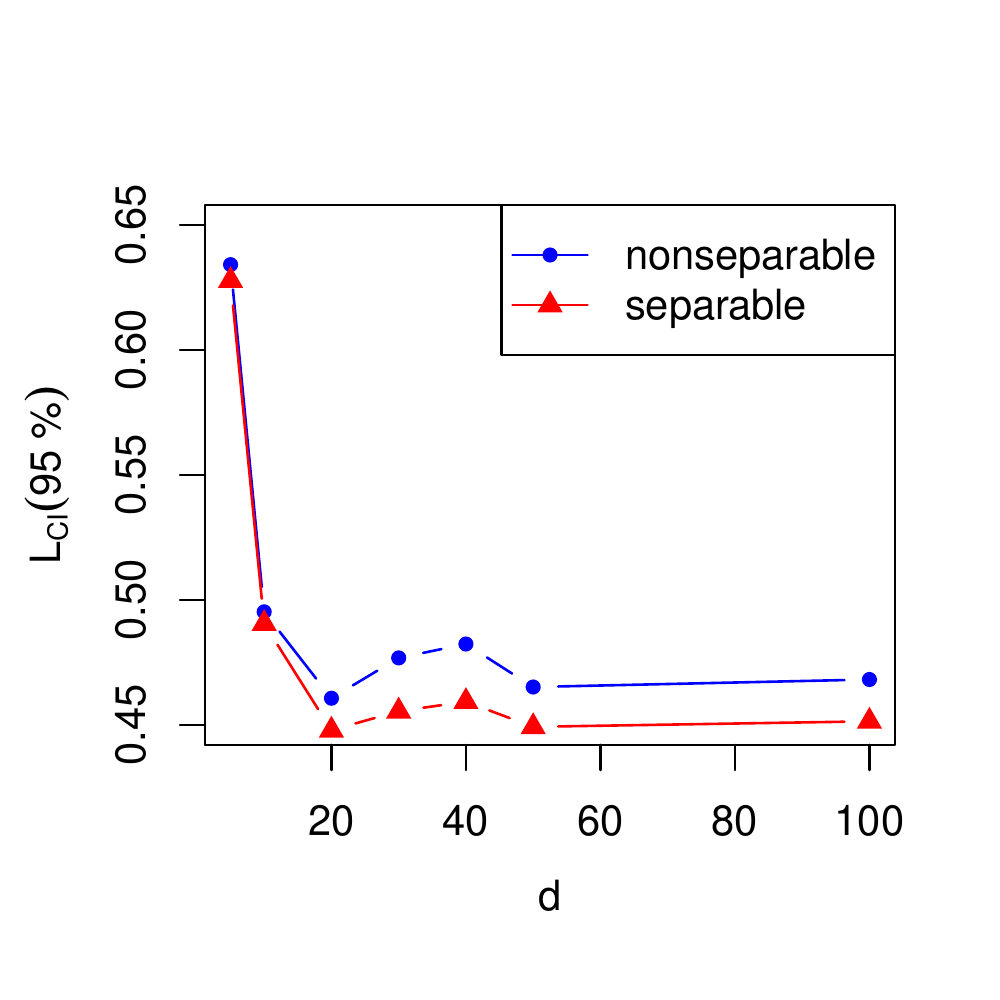}
  \end{tabular}
  \vspace{-.2in}
   \caption{The predictive performance of GOLF process with $d=5,10,20,30,40, 50$ and $100$ factors for Example \ref{eg:compare_diff_factors}, when the true number of factor is $d_{real}=100$ in generating the data. The nonseparable kernel with distinct kernel parameters is assumed to generate the data in the first row of panels, and the separable kernel with the same kernel parameter of each factor process is used for simulation in the second row of panels. The blue curves and red curves denote the performance by the GOLF processes with the different kernel parameters and the same kernel parameter, respectively. In the left panels, the solid curves denote the RMSE for predicting the (noisy) observations, and the dashed curve denotes the RMSE for predicting the mean of the observations. The proportions of observations covered in the $95\%$ predictive interval and the average length of the predictive interval are graphed in the middle and right panels, respectively.  }
   
   
\label{fig:pred_compare_diff_factors_d_real}

\end{figure}

\section{Additional results for real applications in Section \ref{subsec:real_spatial_data}}
\label{sec:supp_real_data}

In this section, we include additional results for GOLF processes predicting the missing values of the temperature data set discussed in \cite{heaton2019case}. We show the details of 5 different configurations of GOLF processes, where the result reported in the main body of the article is the configuration 1. For all the configurations, the proportion of the burn-in samples is $20\%$.
We use the normal distribution centered on the previous values as  the proposal distribution of the logarithm of the inverse range parameters and logarithm of the nugget parameters. For the logarithm of the inverse range parameters of the factor loading matrix,  the standard deviation of the proposal distribution is $40/n_1$. For the logarithm of the inverse range parameters and the nugget parameters of the factor processes, the standard deviation of the posterior distribution is set to be $40/n_2$.




\begin{table}[h]
	  \centering
	{\begin{tabular}{lccccc} \hline
		  & sample size & system & initial $Y^*_{v,i}$& initial  $log(\beta_0)$ & initial $\log(\beta_l)$   \\ \hline
		Conf. 1 & 6000 & Mac & mean  at each latitude & 3 &0\\	
		Conf. 2 & 6000 & Win & mean  at each latitude & 3 &0\\
		Conf. 3 & 40000 & Mac & mean  at each latitude & 3 &0\\
		Conf. 4 & 40000 & Mac & overall mean $+$ noise & 3 &0\\
		Conf. 5 & 40000 & Mac & mean  at each latitude & Unif[-1,1] &Unif[-1,1]\\

		\hline
	\end{tabular}
	}
		 \caption{Detailed settings of 5 different configurations of GOLF processes for the data set in \cite{heaton2019case}. The number of samples and the computing system are shown in the second column and third column, respectively. The choice of the initial values of the missing data is given in the fourth column, using either the mean of the observations at each latitude or overall mean of the observations with a small random Gaussian noise (with standard deviation being $0.1$ times of the standard deviation of the observations). The initial values of the logarithm of the inverse range parameters are either chosen to be a fixed value or randomly sampled from the uniform distribution, shown in columns 5-6.  }
\label{tab:setting_diff_conf}
	\end{table}

	  \begin{table}[h]
	  \centering
	{\begin{tabular}{lcccc} \hline
		Methods  & RMSE & $P_{CI}(95\%)$ & ${L_{CI}(95\%)}$ \\ \hline
		Configuration 1 & 1.46 & 0.92 &4.95\\
		Configuration 2 & 1.50 & 0.91 &4.92\\
		Configuration 3 & 1.44 & 0.94 &7.70\\
		Configuration 4 & 1.48 & 0.94 &7.75\\
		Configuration 5 & 1.51 & 0.93 &5.16\\

		\hline
	\end{tabular}
	}
		 \caption{Predictive performance of 5 different implementations for the data set in \cite{heaton2019case}.}
\label{tab:prediction_diff_conf}
	\end{table}
	
	The details of 5 configurations are given in Table \ref{tab:setting_diff_conf}. The predictive RMSE, $P_{CI}(95\%)$ and ${L_{CI}(95\%)}$ of the 5 configurations are given in Table \ref{tab:prediction_diff_conf}. The predictive RMSE is similar for all 5 configurations. Increasing the posterior sample size seems to slightly increase the proportion of the samples contained in the $95\%$ predictive interval.

	\begin{figure}[H]
\centering
  \begin{tabular}{c}
	\includegraphics[height=.4\textwidth,width=1\textwidth]{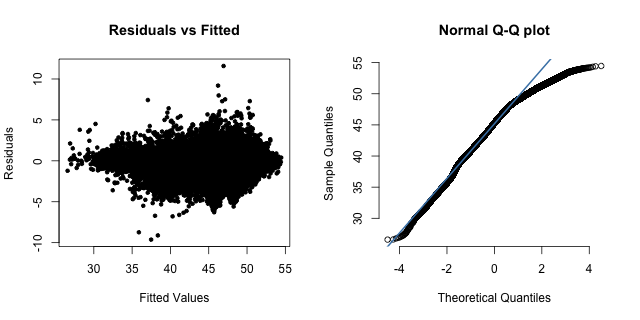}
  \end{tabular}
   \caption{Diagnostic plots of the GOLF processes for the data set in \cite{heaton2019case}.}

   
\label{fig:diagnostic_plot_real_eg_1}

\end{figure}

	The fitted values from the GOLF processes in configuration 1 against the residuals and the normal Q-Q plot  are graphed in the left panel and the right panel in Figure \ref{fig:diagnostic_plot_real_eg_1}, respectively. The Q-Q plot indicates the fitted values are slightly left-skewed and slightly under-dispersed.

\end{document}